\pdfoutput=1 

\documentclass[12pt]{article}
\usepackage{amsmath}
\usepackage{graphicx}
\usepackage{enumerate}
\usepackage{url} 
\usepackage{float}
\usepackage[utf8]{inputenc} 
\usepackage[T1]{fontenc}    
\usepackage{hyperref}       
\usepackage{url}            
\usepackage{booktabs}       
\usepackage{amsfonts}       
\usepackage{nicefrac}       
\usepackage{microtype}      
\usepackage{lipsum,  amsthm}
\usepackage[ruled,vlined]{algorithm2e}
\usepackage{wrapfig}
\usepackage{graphicx}
\usepackage{authblk}
\usepackage[backend=biber,style=numeric]{biblatex}
\usepackage{blindtext}
\usepackage{caption}

\newtheorem{thm}{Theorem}
\newtheorem{lemma}{Lemma}
\newtheorem*{lemma*}{Lemma}
\newtheorem*{thm*}{Theorem}

\DeclareMathOperator*{\argmin}{arg\,min}

\newcommand{\blind}{1}

\begin{document}

\def\spacingset#1{\renewcommand{\baselinestretch}%
{#1}\small\normalsize} \spacingset{1}


\title{\bf Gradient-Based Markov Chain Monte Carlo for Bayesian Inference With Non-Differentiable Priors}
  
\if1\blind
{
\author[1]{Jacob Vorstrup Goldman$^*$}
\author[2]{Torben Sell$^+$\thanks{Joint first authors. JVG acknowledges financial support from an EPSRC Doctoral Training Award. TS acknowledges financial support from the Cantab Capital Institute for the Mathematics of Information.\\+: Corresponding author, e-mail: \href{torben.sell@cantab.net}{torben.sell@cantab.net}\\This manuscript has been accepted for publication by the Journal of the American Statistical Association, \href{https://www.tandfonline.com/toc/uasa20/current}{https://www.tandfonline.com/toc/uasa20/current}.}}
\author[1]{Sumeetpal Sidhu Singh}
    
    \affil[1]{Signal Processing and Communications Laboratory,
    Department of Engineering, University of Cambridge, UK}
    \affil[2]{    Department of Pure Mathematics and Mathematical Statistics, University of Cambridge, UK}
  \maketitle
} \fi

\if0\blind
{
  \author{}
  \maketitle
  \bigskip
} \fi

\begin{abstract}
The use of non-differentiable priors in Bayesian statistics has become increasingly popular, in particular in Bayesian imaging analysis. Current state of the art methods are approximate in the sense that they replace the posterior with a smooth approximation via Moreau-Yosida envelopes, and apply gradient-based discretized diffusions to sample from the resulting distribution. We characterize the error of the Moreau-Yosida approximation and propose a novel implementation using underdamped Langevin dynamics. In misson-critical cases, however, replacing the posterior with an approximation may not be a viable option. Instead, we show that Piecewise-Deterministic Markov Processes (PDMP) can be utilized for exact posterior inference from distributions satisfying almost everywhere differentiability. Furthermore, in contrast with diffusion-based methods, the suggested PDMP-based samplers place no assumptions on the prior shape, nor require access to a computationally cheap proximal operator, and consequently have a much broader scope of application. Through detailed numerical examples, including a non-differentiable circular distribution and a non-convex genomics model, we elucidate the relative strengths of these sampling methods on problems of moderate to high dimensions, underlining the benefits of PDMP-based methods when accurate sampling is decisive. 
\end{abstract}

\noindent%

\spacingset{1.45} 

\section{Introduction}
Estimating and quantifying the uncertainty of the parameters in a Bayesian statistical model often involves intractable, high-dimensional integrals. One of the most widely applied methods to estimate these integrals is Markov Chain Monte Carlo (MCMC), which involves simulating a Markov chain that has the posterior distribution as its invariant, and subsequently estimating quantities of interest from the resulting trajectory. More recently, algorithms that operate directly in continuous-time have attracted significant attention \cite{bierkens2019zig, bouchard2018bouncy}. These samplers are irreversible, which has been shown to increase mixing speeds \cite{neal2004improving}, and allow for exact subsampling of data by exploiting the factor structure of product-likelihoods, completely avoiding the bias of stochastic gradients; this subsampling operation can furthermore be done at constant cost in common cases \cite{bierkens2019zig}. This paper discusses continuous-time algorithms for simulating values, or sampling, from the posterior of Bayesian problems with only almost everywhere differentiable posteriors.

This broad class of distributions includes all log-concave posteriors, as well as posteriors that arise from a log-concave and non-differentiable prior, and a differentiable likelihood. This latter class of distributions have a long history in convex optimization, where the non-differentiable priors are used to ensure existence and regularity of solutions \cite{boyd2004convex}. The most widespread uses of non-differentiable priors have been in the image analysis literature, for example in denoising \cite{rudin1992nonlinear, chambolle2010introduction}, deblurring \cite{babacan2008variational, beck2009fast}, multiframe super resolution \cite{farsiu2004fast} and compressed sensing \cite{candes2006robust, babacan2009bayesian}. Outside of image analysis, the Laplace, or double exponential, prior is used in Lasso regression and its Bayesian counterpart \cite{tibshirani1996regression, park2008bayesian}, in sparse regularization via the Bernoulli-Laplace prior \cite{chaari2013sparse}, and in source localization \cite{fernandez2004circular}. More sophisticated scale-mixture priors are introduced in \cite{carvalho2010horseshoe, bhattacharya2015dirichlet, hosseini2019two}. For low-rank matrix completion, the nuclear-norm prior over the singular values is a convex envelope of the rank function, and non-differentiable \cite{babacan2012sparse, koltchinskii2011nuclear}. 
To be able to utilize gradient-based methods for these kind of problems, the convex optimization community utilizes the proximal operator and the related Moreau-Yosida envelope (e.g., \cite{nitanda2014stochastic, bauschke2011convex}, see \cite{parikh2014proximal} for a review) to smoothen the convex prior. The resulting regularized prior function is known as the Moureau-Yosida envelope (MYE) prior, and is everywhere continuously differentiable, in particular, the gradient is just a linear function of the proximal operator. The closeness of the MYE to the true underlying function is determined by a single \textit{envelope tightness parameter} $\lambda$, and this parameter plays a key role in what follows. We will generically denote a posterior density that incorporates a MYE prior as $\pi^\lambda.$ 

Recently, MCMC algorithms and computational resources have matured to a point where it is now feasible to carry out full posterior inference for high-dimensional instances of these models, rather than just maximum a posteriori (MAP) estimation as is done in optimization. Since its introduction to the sampling literature in the seminal paper of \cite{pereyra2016proximal}, the proximal operator and the MYE have been used to propose new gradient-based sampling methods for non-differentiable distributions \cite{durmus2018efficient, chaari2016hamiltonian, salim2019stochastic}. The use of these operators in sampling algorithms is, however, not without complications. As the gradient of the MYE is $1/\lambda$-Lipschitz continuous \cite[Proposition 3.1]{durmus2018efficient} and simultaneously the accuracy of the posterior approximation is increased by decreasing $\lambda$, the step-sizes one may take in the algorithms of \cite{pereyra2016proximal} or \cite{salim2019stochastic} need to scale at order $O(\lambda)$ to ensure that the invariant distribution of the diffusion process is close to the target $\pi^\lambda$. The unadjusted diffusions also introduce bias from both their discretization and the use of the MYE to approximate the target density, yet it is only possible to adjust for the bias via Metropolis-Hastings corrections if the proposal density is available in closed form. This correction also comes at the cost of slower mixing. In mission-critical cases, such as calculating the probability of disease via MRIs, see for example \cite{le1995blount}, or defect detection of large-scale civil infrastructure via visual sensing technologies \cite{koch2015review}, it follows that end users will be interested in as exact sampling as possible to most accurately quantify the resulting uncertainty of estimates. Furthermore, Langevin-based methods are also less robust under heavy-tailed and weakly log-concave distributions in general \cite{livingstone2019robustness}, and these distributions are widely applied in both imaging and Bayesian statistics. In addition, the computational cost of calculating the proximal operator is also completely model dependent, with effective algorithms only existing in particular cases. Finally, the use of data subsampling via stochastic gradients \cite{welling2011bayesian, chen2014stochastic} has seen widespread adoption. There is, however, no theoretical understanding as of yet on the combination of stochastic gradients and MYE-based priors, as the two approximation methods induce competing sources of error. This limits the application of MYE-based methods to cases where the data is of manageable size, typically images, rather than the tall and wide data sets faced in the contemporary big data regime.

\subsection{Contributions}

The contributions of this paper are:
\begin{itemize}
    \item The availability of the gradient through the proximal operator has in particular allowed for the use of Langevin dynamics \cite{roberts1996exponential, roberts1998optimal}, where the gradient of the MYE of the log-posterior is applied as the drift-function in a discretized stochastic differential equation (SDE). If this SDE is run unadjusted, i.e. without a Metropolis-Hastings acceptance step, we show in Theorem \ref{thm_expectation_bounds} that one can asymptotically compute good approximations to the true expectations if the envelope tightness parameter $\lambda$ is chosen small enough. 
    \item A common thread for the Langevin-based algorithms presented in \cite{pereyra2016proximal, durmus2018efficient, salim2019stochastic} is that they all use first-order methods to discretize the Langevin SDE. As an extension to the existing literature, we show that the second-order unadjusted \emph{Underdamped} Langevin Dynamics can be used to target the smoothed posterior. Numerically, this can lead to better mixing \ref{ex:100D_laplace}. These dynamics require an extra momentum variable but come with provably faster convergence properties \cite{ma2019there} at limited extra computational effort. 
    \item We show that the recent class of continuous-time sampling algorithms \cite{bouchard2018bouncy, bierkens2019zig, bou2017randomized}, all based on Piecewise-deterministic Markov Processes (PDMP, originally introduced in \cite{davis1984piecewise}, for a thorough overview see \cite{jacobsen2006point}), can exactly sample from our class of non-differentiable posteriors by only using the gradient at points of differentiability. It is therefore unnecessary to calculate the potentially costly MYE or proximal operator, and the samples are asymptotically exact draws from the correct posterior $\pi$ rather than the approximate $\pi^\lambda$. 
    \item We provide numerical comparisons of these samplers in moderate to high-dimensional problems from Bayesian statistics and imaging, evaluate and compare the performance of the algorithms, and give clear recommendations for end users applying non-differentiable priors. Notably, we carry out exact Bayesian inference on a non-differentiable circular distribution, and in a non-convex genomics model where proximal operators are inapplicable, recovering significant gene expressions in the latter. In the Supplementary Material \ref{appendixB}, we also provide new insight on the performance of PDMP methods in a challenging Bayesian imaging model in high dimension, and a matrix denoising example using Hamiltonian dynamics \cite{vanetti2017piecewise}.
\end{itemize}

The paper is organized as follows: In Section \ref{sec_prox_MYE} we introduce the proximal operator and the Moreau-Yosida Envelope (MYE), as well as giving an error estimate of the smooth target. In Section \ref{sampling_section}, we review the existing Langevin-based methods for a MYE-smoothed posterior, introduce a new second-order method, and discuss how to use PDMPs to exactly sampling the target. Section \ref{examples} compares the algorithms on a variety of examples, and Section \ref{conclusion} concludes. Long proofs can be found in Supplementary Material \ref{appendixA}, additional numerical experiments in Supplementary Material \ref{appendixB}.

\subsection{Frequently Used Notation}
We consider a random variable of interest with range $x \in \mathcal X$ where typically $\mathcal X = \mathbb R^n$. We define the prior density $\pi_0(x)\in C^0_{\text{a.e.}}(\mathbb R^n,\mathbb R)$, where $C^0_{\text{a.e.}}$ is the set of a.e.-differentiable continuous functions, functions of this class will also be denoted non-differentiable. Given data $y \in \mathbb R^m$ for some $m > 1$ distributed conditionally on $x$ according to $p(y|x)$, let the likelihood be $x \mapsto\mathcal L(x) = p(y|x) \in C^1(\mathbb R^n,\mathbb R);$ furthermore define $\ell(x)=\log \mathcal L (x),$ the log-likelihood. Since data is fixed, we subdue it in the notation from now on. By Bayes' theorem, the posterior distribution of $x$ given data is then $\pi(x) \propto \mathcal L(x) \pi_0(x) \in C^0_{\text{a.e.}}(\mathbb R^n,\mathbb R)$. It will be convenient to work entirely in log-space, we therefore define the \emph{posterior potential} up to a constant as $U(x) \equiv -\log \pi(x) =-\ell(x) - \log \pi_0(x)$. By $\Vert \cdot \Vert$ we denote the Euclidean $\ell^2$ norm, by $\lVert \cdot\rVert_1$ the $\ell^1$-Norm. The constant $\lambda$ will always indicate the precision of the Moreau-Yosida envelope. The domain of a function $g$ is denoted $\textbf{dom}(g)$.

\section{Proximal Operators and the Moreau-Yosida Envelope\label{sec_prox_MYE}}

\subsection{Proximal Operators}
We will now describe, and collect a few useful results about, the proximal operator. These allow to simplify the gradient evaluation of the approximations to non-differentiable posteriors discussed in the next subsection. The proximal operator $\text{prox}_g^\lambda$ is defined for any convex function $g: \mathbb{R}^n \rightarrow \mathbb{R}$ and parameter $\lambda$ as the optimization problem
\begin{align}
\text{prox}_g^\lambda (x) = \underset{u}{\argmin} \Big \{g(u) + \frac{1}{2\lambda} \Vert x - u \Vert^2 \Big \}.
\end{align}
The proximal operator is in fact a generalization of the Euclidean projection operator: let $g(x)=0$ if $x\in A$, $g(x)=\infty$ otherwise, for some convex set $A \in \mathbb R^n$, then the resulting proximal operator is $\text{prox}_g^\lambda(x) = \text{proj}_A(x)$ \cite{parikh2014proximal}.
In addition, if $g \in C^1(\mathbb R^d)$, then the prox$_g^\lambda$ operator satisfies 
\begin{align}\label{prox_op_1}
p=\text{prox}_g^\lambda (x)~\iff~x-p=\lambda \nabla g(p),
\end{align}
or, if $g$ is convex but not differentiable, it still satisfies $p=\text{prox}_g^\lambda (x)~\iff~x-p\in \lambda\partial g(p)$,
where $\partial g(p)$ denotes the subdifferential relation $\partial g(p)=\{u\in\mathbb R^n:\forall y\in\mathbb R^n,(y-p)^Tu+g(p)\leq g(y)\}$.

\subsection{The Moreau-Yosida envelope}
One can approximate a non-differentiable and convex function $g$ with the \emph{Moreau-Yosida envelope} (MYE) which is defined by
\begin{align}\label{eq:mye_opt}
    g^\lambda(x)=\inf_z\left[g(z)+\frac{1}{2\lambda}\lVert x-z\rVert^2\right].
\end{align}
The MYE is again a convex function by convexity of the infimum, and as the positivity of the quadratic term preserves minima. If $g$ furthermore is $L$-Lipschitz continuous, this envelope is close to the original function, as the following theorem from \cite[Proposition 3.4 with $\lambda=1/r$]{hosseini2019nonsmooth} shows:
\begin{thm}\label{thm_MYE_closeness}
Let $g:\mathcal X\rightarrow ]-\infty,\infty]$ be a proper lower semi-continuous convex function, and $L$-Lipschitz. Let $\lambda>0$. Then for any $x\in \textbf{dom}(g)$,
\begin{align*}
0\leq g(x)-g^\lambda(x)\leq \frac{L^2\lambda}{2}.
\end{align*}\end{thm}
This characterizes the trade-off between the precision parameter $\lambda$ and the Lipschitz constant of the gradient of $g$, such that $\lambda$ should be chosen of the order $\mathcal O(L^{-2})$ to achieve tightness of the approximation. We note that the upper bound is often achieved, e.g. for the example in Figure \ref{fig:mye_test} for $\lambda=0.25$ we have a Lipschitz constant $L=1$, and for any $x$ with $\lvert x\rvert>\lambda$ we can easily check that $g(x)-g^\lambda(x)=\frac\lambda2$. Additionally, the MYE is differentiable everywhere, and we can compute the derivative using the proximal operator defined in the previous section \cite[Theorem 12.30]{bauschke2011convex}:
\begin{thm}\label{theorem:log_gradient_prox}
Let $g:\mathcal X\rightarrow ]-\infty,\infty]$ be a proper lower semi-continuous convex function, and let $\lambda>0$. Then $g^\lambda$ is Fr\'echet differentiable, and its gradient is $1/\lambda$-Lipschitz continuous, given by
\[
\nabla g^\lambda(x)=\frac1\lambda (x-\emph{prox}_g^\lambda (x)).
\]
\end{thm}
Rearranging the above equation reveals that iterative application of the proximal operator just corresponds to gradient descent of the MYE-smoothed version of $g$. In Figure \ref{fig:mye_test}, we provide a simple visual aid that illustrates the behavior of both the MYE and proximal operator and the respective densities in a basic example, $g(x) = |x|$ such that $g \in C^0$. The resulting MYE in this case is the Huber loss function. In general, we do not have access to closed-form solutions to either the MYE or the proximal operator, and the solution to either can be quite computationally expensive. 

\begin{figure}
	\includegraphics[scale=0.5]{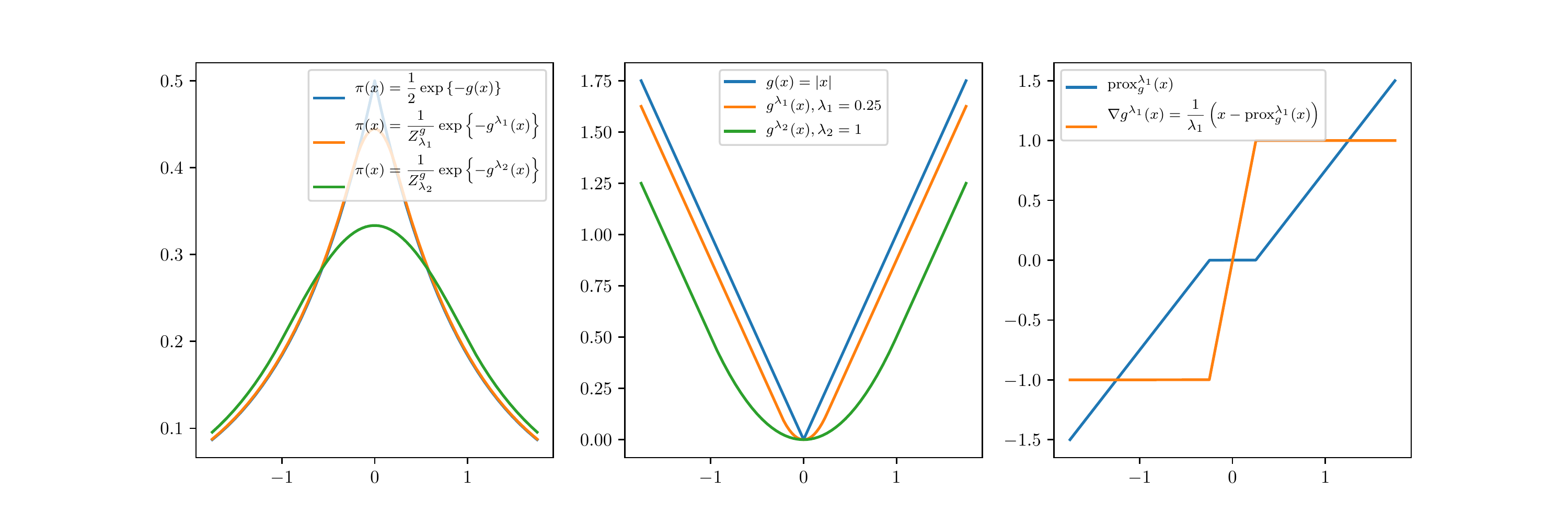}
	\caption{Left: The Laplace distribution and two MYE enveloped versions with $\lambda_1=0.25$ and $\lambda_2=1$, respectively. The normalizing constant of the MYE-adjusted density is $Z^g_\lambda = \int_{\mathbb R} \exp\{-g^\lambda(x) \} dx.$ Center: Plot of the log-densities. Right: The proximal operator of $g$ and the corresponding gradient of the MYE $g^\lambda(x)$.}
	\label{fig:mye_test}
\end{figure}

Having recalled these basic facts about the proximal operator, we now provide an easily verifiable bound on the precision achievable when using the MYE in sampling algorithms: Theorem \ref{thm_expectation_bounds} quantifies the error obtained if we compute (exactly) an expectation with respect to the smoothed target distribution versus the true (non-differentiable) distribution in the sense that for suitably regular $f:\mathcal X\rightarrow \mathbb R$, $\mathbb E_{\pi^\lambda}(f)\approx\mathbb E_\pi(f)$. The benefit of the first inequality is that we can easily verify the right-hand side numerically and thus get an error bound estimate.

\begin{thm}\label{thm_expectation_bounds}
Let $g=-\log \pi$ be the negative logarithm of a probability density function, with $g$ being a proper lower semi-continuous convex function, and $L$-Lipschitz. Let $g^\lambda$ be the Moreau-Yosida envelope to $g$, and let $\pi^\lambda(x)=\exp(-g^\lambda(x))/(\int\exp(-g^\lambda(z))dz)$ be a probability density function. Then for any $\pi$- and $\pi^\lambda$-integrable $f:\mathcal X\rightarrow \mathbb R$:
\begin{align}
    &|\mathbb E_{\pi^\lambda}(f)-\mathbb E_\pi(f)|\leq(\exp(L^2\lambda)-1)\mathbb E_{\pi^\lambda}(|f|)\\
    &|\mathbb E_{\pi^\lambda}(f)-\mathbb E_\pi(f)|\leq(\exp(L^2\lambda)-1)\mathbb E_\pi(|f|).
\end{align}
The same inequalities hold if $g=g_1+g_2$ with a convex and Lipschitz-continuous $g_1$ and a differentiable (but not necessarily Lipschitz-continuous) $g_2$: In that case, one takes the MYE of $g_1$ only, resulting in the approximate pdf $\pi^\lambda(x)=\exp(-g_1^\lambda(x)-g_2(x))/(\int\exp(-g_1^\lambda(z)-g_2(z))dz)$.
\end{thm}
\begin{proof}
See Supplementary Material \ref{proof:thm_expectation_bounds}.
\end{proof}

Choosing $f(x)=sgn(\pi^\lambda(x)-\pi(x))$ in this theorem allows to recover proposition 3.1 in \cite{durmus2018efficient} by only considering a one-sided inequality in Equation \ref{proof_ineq:twobounds} in the proof of the theorem, and can thus be viewed as a generalization thereof. Theorem \ref{thm_expectation_bounds} shows that for exact integral estimates, one wants to pick $\lambda$ as small as possible. This, however, does come at a cost, as the gradients of the MYE approximated target grow as $\lambda\rightarrow0$ such that one needs to take a smaller step size in the diffusion algorithms discussed in the next section. The relation between different approximation parameters and the size of the gradient is given by the next lemma.

\begin{lemma}\label{thm_gradient_bounds}
Let $g:\mathcal X\rightarrow ]-\infty,\infty]$ be a proper lower semi-continuous convex function, and let $0<\lambda_1\leq\lambda_2$; then for the corresponding Moreau-Yosida envelopes $g^{\lambda_1}$ and $g^{\lambda_2}$, we have
\begin{align*}
    \lVert \nabla g^{\lambda_1}(x)\rVert\geq\lVert\nabla g^{\lambda_2}(x)\rVert\quad\forall x\in\mathcal X.
\end{align*}
\end{lemma}
\begin{proof}
See Supplementary Material \ref{proof:thm_gradient_bounds}.
\end{proof}

\section{Continuous-time stochastic processes for sampling from non-differentiable posteriors}
\label{sampling_section}

We now discuss ways of sampling from posteriors that are non-differentiable. We will firstly discuss algorithms based on the Langevin equations and proximal operators, and introduce another Langevin type method using underdamped, or second-order, Langevin dynamics. We will then discuss how samplers based on Piecewise-deterministic Markov Processes can be used in the above mentioned setting. 

\subsection{Langevin dynamics}
\subsubsection{Overdamped Langevin Dynamics}
The overdamped Langevin equation is defined as the stochastic differential equation given by
\begin{align}
\label{overdamped_langevin}
    dx_t &= \nabla U(x_t)dt+\sqrt{2}dB_t,
\end{align}
where $U(x)$ is the log-posterior and $B_t$ is a $n$-dimensional Brownian motion. The resulting invariant distribution of the semi-group associated to the Langevin equation is under smoothness assumptions proportional to $\exp(-U(x))$. This implies that equation \eqref{overdamped_langevin} can be simulated according to some discretization scheme, typically Euler-Maruyama, to generate samples that are distributed according to $\pi$. The discretization of \eqref{overdamped_langevin} leads to bias that is typically corrected with a Metropolis-Hastings (M-H) step, leading to the popular Metropolis-adjusted Langevin Algorithm (MALA) \cite{roberts1996exponential}. 
Large-scale, data-intensive models has lead to a notable increase in interest in unadjusted Langevin algorithms (ULA) \cite{welling2011bayesian}, as no accept/reject step is applied. In this case, mixing is generally improved \cite{durmus2018efficient} and gradient evaluations are not wasted on rejected proposals. The step-size is also not forced to comply with theoretically-optimal acceptance rates \cite{roberts1998optimal}, rather, the step-size is chosen to be of the order of $\lambda$ \cite{durmus2018efficient}.

For a general non-differentiable distribution $\pi$ one may, as proposed in \cite{pereyra2016proximal} and further studied in \cite{durmus2018efficient}, target the MYE-smoothed version $\pi^\lambda$ instead. In particular, the gradient of \eqref{overdamped_langevin} is split into a likelihood derivative and an evaluation of the proximal operator via Theorem \ref{theorem:log_gradient_prox}:
\begin{align}\label{no_split_posterior}
\nabla \log \pi^\lambda(x) &= \nabla_x \ell(x) + \frac{1}{\lambda} \Big (x - \text{prox}_{\log \pi_0}^\lambda(x) \Big ) =: -\nabla U_\lambda(x).
\end{align}
The resulting sampler is known as proximal MALA (pMALA) or MY-ULA, depending on whether or not a Metropolis-Hastings step is included. In \cite{vargas2019accelerating}, the authors use stabilized explicit integrators to simulate the diffusion \ref{overdamped_langevin}, which allows accelerated sampling from the posterior. Such stabilized explicit integrators are particularly strong for stiff PDEs, and the SK-ROCK algorithm from \cite{vargas2019accelerating} is included in our numerical comparisons where applicable.

\subsubsection{Underdamped Langevin dynamics}
The dynamics of the overdamped Langevin equation are characterized by reversible, diffusive behavior. A generic way to alleviate these backtracking tendencies over short time-scales is to introduce persistence in the trajectories via a notion of velocity. We therefore augment our target space $\mathcal X$ with $\mathbb R^n$, and on this space let $v$ be a $n$-dimensional vector of velocities drawn from $\mathcal{N}(0, u I_n).$ 
To this end, we now consider the underdamped Langevin Diffusion on $(x,v)\in\mathbb R^{n\times n}$
\begin{align}
    dx_t &= v_tdt\\
    dv_t &= -\gamma v_tdt+u\nabla U(x_t)dt+\sqrt{2\gamma u}dB_t,
\end{align}
for $u > 0$ and $\gamma > 0$ a speed parameter.
Under mild conditions, this process has an invariant distribution proportional to $\exp(-U(x)-\lVert v\rVert_2^2/(2u))$, and thus the marginal distribution of $x$ is the desired target distribution $\pi(x)\propto\exp(-U(x))$. A second-order discretization scheme is provided in the Supplementary Material \ref{appendix:discretisation}. 
Recent theoretical advances have been made to elucidate non-asymptotic properties and the speed of convergence in probability metrics (KL-divergence, Wasserstein distance, etc.) for various formulations of Langevin-based algorithms, see e.g. \cite{dalalyan2017theoretical, cheng2017underdamped, wibisono2019proximal}. In the case where $U$ is $m$-strongly log-concave and Lipschitz-differentiable with parameter $L$, the number of samples required to achieve $\epsilon$ precision in Wasserstein-2 distance scales with $O(\sqrt{n})$, with $n$ the dimension of the model, compared to $O(n)$ for the overdamped dynamics \cite{cheng2017underdamped}. A natural extension of the MY-ULA algorithm is therefore to use the more elaborate dynamics of the underdamped Langevin SDE to explore the smoothed target distribution $\pi^\lambda$:
 \begin{align}\label{ULdynamics}
    dx_t &= v_tdt\\
    dv_t &= -\gamma v_tdt+u\nabla U^\lambda(x_t)dt+\sqrt{2\gamma u}dB_t.
\end{align}
As with MY-ULA, the convergence of a discretization of these dynamics depends on the Lipschitz constant of the gradient \cite{cheng2017underdamped}, which implies that the trade-off between posterior accuracy and mixing speed remains with MY-UULA. In spite of this, the improved dimensional scaling of MY-UULA can provide significant improvements over MY-ULA, see e.g. Example \ref{ex:100D_laplace}.

\subsection{Piecewise-Deterministic Markov Processes}
Below we elucidate the two most popular types of PDMP samplers, the Bouncy Particle Sampler (BPS) (First introduced in statistical physics by \cite{peters2012rejection} and subsequently ported to statistics in \cite{bouchard2018bouncy})
and the Zig-Zag Sampler (ZZS) \cite{bierkens2019zig}. Both are augmented-variable methods that introduce a notion of velocity to accelerate exploration. PDMP-samplers explore the target with persistent, deterministic dynamics, interspersed with direction-changes at random times. For both types of samplers, the key quantity is, similarly to the diffusion-based methods, the gradient of the target density. In contrast, however, the computational cost of running the samplers is determined by bounds on the gradient, and it only enters the dynamics at select times through the jump-process that updates the velocities. To run the samplers under a.e. differentiability, note that the gradient is well-defined everywhere outside a Lebesgue-nullset
\begin{align}\label{A_0}
    A_0 = \left \{x \in \mathcal X \; \vert \; \exists i \text{ such that } \frac{\partial U}{\partial x_i} \text{ does not exist.} \right \};
\end{align}
we will on that subset replace any undefined derivative with zero, detailed descriptions and proofs are given below. We begin with the ZZS, the BPS and its variations follows subsequently. We end this section with a short overview of how to simulate these processes in practice. 

\subsubsection{Zig-Zag Sampler}
We augment our space with a new random variable, which we will denote velocity. Let $v \in \{-1, 1\}^n$, and consider the uniform distribution over velocities $p(v) = \mathcal U(v)$. The continuous-time process $(z_t)_{t \geq 0} = (x_t, v_t)_{t \geq 0}$ associated with the ZZS targets the joint distribution $p(x,v) = \pi(x) \mathcal U(v)$.

Consider an initial value $(x_0, v_0) \sim p_0(x_0, v_0)$. The deterministic flow of the ZZS between events is given as a solution to the ODE $(\dot{x_t}, \dot{v_t}) = (v_t, 0)$, which is just $(x_t, v_t) = (x_0 + v_0 \cdot t, 0)$, indicating that the flow of $x$ over time is a continuous but non-differentiable process in $x$. For each individual dimension of the problem, we associate an inhomogeneous Poisson process (iPP), and the rate at which a jump, a change in velocity, occurs for dimension $i$ is given by
\begin{align*}
    \rho^i_{ZZ}(t) := \rho_{ZZ}^i(t; x, v) = \max \left \{0, \frac{\partial}{\partial x_i} U(x + v\cdot t) \cdot v_i \right \}
\end{align*}
The rate at which jumps occur globally for the sampler is subsequently just $\rho_{ZZ}(t) = \sum_{i=1}^n \rho_{ZZ}^i(t)$. In the case of prior non-differentiability, define
\begin{align*}
    B_x = \left \{ 1,2, \ldots,n \; \vert \; \frac{\partial U}{\partial x_i}(x) \text{ does not exist.} \right \}, \quad \forall x \in A_0
\end{align*}
with $A_0$ as in equation \ref{A_0}, and subsequently set $\{ \partial U / \partial x_j(x)\}_{j \in B_x}$ equal to zero.
Intuitively, if the derivative of $U(x)$ with respect to a coordinate $i$ is of the same sign as the velocity $v_i$, then it implies that the particle is entering the tail of the density. As this accumulation continues, one of the rates will generate an event in an attempt to force a return to regions of higher probability.
The rate that generated the event and the new event time $\tau$ are in practice found by
\begin{align*}
    j^* = \argmin_{j=1,2,\ldots, n} \tau_j, \quad \tau = \tau_{j^*}
\end{align*}
respectively, and for the particular dimension $j^*$, the flip operator $\mathcal F_{j^*} v$ is applied:
\[
    \mathcal F_{j^*} v = \begin{cases}
                        v_i = -1 \cdot v_i, \quad&\text{if}\quad  i = j^* \\
                        v_k = v_k, \quad &\text{else,}
                        \end{cases}
\]
which trivially preserves the volume of the extended target. In this sense the ZZS is naturally localized across all dimensions, as each dimension depends only on the variables that interacts with the derivative with respect to that coordinate. The infinitesimal generator of the Markov process associated with the ZZS is
\begin{equation}\label{eq:zz}  
\mathcal L_{ZZ} f(x,v) = \langle \nabla_x f(x), v \rangle + \sum_{i=1}^n \rho_{ZZ}^i(t) [f(x, \mathcal F_i v) - f(x, v)].
\end{equation}

\begin{lemma}\label{lemma:zz_invariance}
Consider a distribution $\pi(x)\mathcal U(v)$ that is differentiable in $x$ outside of $A_0$. If $A_0$ is a Lebesgue null-set, the Zig-Zag process with generator given in Equation (\ref{eq:zz}) has invariant distribution $\pi(x) \mathcal U(v)$.
\end{lemma}
\begin{proof}
See Supplementary Material \ref{proof:zz_invariance}.
\end{proof}

\subsubsection{Bouncy Particle Sampler} Given our variable of interest $x \in \mathbb R^n$, we again augment the state space with an additional $n$-dimensional component $v$ and assume that in stationarity $v \sim N(0, I_n)$. With the same flow as the zig-zag sampler, the particle continues along a trajectory until an event occurs, and the rate at which these jumps of the velocity occur is determined by the rate function of a single iPP:
\[
\rho_{BPS}(t) := \rho(t; x, v) = \max \{0, \langle v, \nabla U(x + v\cdot t) \rangle \},
\]
where we recall that $U(x)$ is the negative log-probability of the posterior distribution. Whenever an event occurs, all velocities are updated globally via the deterministic transition operator 
\[
\mathfrak R_x v = v - 2 \frac{\langle v, \nabla U(x) \rangle}{\Vert \nabla U(x) \Vert^2 }\nabla U(x),
\]
which corresponds to a reflection in the hyperplane orthogonal to the gradient at $x.$ As both the rate and the reflection operator is undefined on $A_0$, we set $\nabla U(x) = 0, \forall x \in A_0$. In some cases the BPS has been observed to be reducible \cite{deligiannidis2019exponential}. To avoid this degenerate behaviour, velocity refreshments are introduced: at some rate $\phi > 0$ we draw a new velocity from the stationary distribution $\mathcal N(v; 0, I_n)$; we denote this independent kernel by $Q(v')$. The infinitesimal generator of the BPS with refreshment is then
\begin{equation}\label{eq:bps}
\mathcal L_{BPS} f(x,v) = \langle \nabla_x f(x), v \rangle + \rho_{BPS} (x,v) [f(x, \mathfrak R_x v) - f(x, v)] + \phi \int_{\mathbb R^{n}} [f(x, v') - f(x,v)] Q(v') dv'
\end{equation}
for all $f \in D(\mathcal L)$, the domain of $\mathcal L$. 
\begin{lemma}\label{lemma:bps_invariance}
Consider a distribution $\pi(x)\mathcal \mathcal N(v; 0, I_n)$ that is differentiable in $x$ outside of $A_0$. If $A_0$ is a Lebesgue null-set, the Bouncy Particle Sampler with generator given in Equation (\ref{eq:zz}) has invariant distribution $\pi(x) \mathcal N(v; 0, I_n)$.
\end{lemma}
\begin{proof}
See Supplementary Material \ref{proof:bps_invariance}.
\end{proof}
For high-dimensional problems, the BPS will without refreshments remain on a single contour of $U$ and another contour of $\log p(v)$ since $\Vert v \Vert = c$ for some $ c > 0$, in fact, the limiting dynamics correspond to those of the randomized Hamiltonian Monte Carlo algorithm, see \cite{deligiannidis2018randomized}. To circumvent this issue, it is in general necessary to exploit some factor structure of the posterior potential:
\[
U(x) = \sum_{i} U_i(x), 
\]
where each $U_i$ can depend on any number of dimensions of $x$. The BPS extends easily to this case, as the individual factors can run as fully local bouncy particle samplers. With subscript $i$ on variables denoting restriction to the components in factor $i$, the event-times are now determined by local rate functions $\rho^i_{BPS}(t) = \max \{ 0, \langle v_i, \nabla U_i(x_i + v_i\cdot t) \rangle \}$, and the actual reflection event time is, as it is for the ZZS, just the minimum over all event times, $\tau_r = \min_{i} \tau_i$. Accordingly, the reflection operator for each factor instead only uses the the gradient of the reflecting factor, and only updates the corresponding subset $v_i$ of the velocity vector via 
\[
\mathfrak R_x^i v_i = v_i - 2 \frac{\langle v_i, \nabla U_i(x_i) \rangle}{\Vert \nabla U_i(x_i) \Vert^2 }\nabla U_i(x_i).
\]
We denote the factorized version the \emph{local BPS} (LBPS), in contrast with the standard global version. In the supplement, Section \ref{supp:hamiltonian}, we include an extended version with Hamiltonian dynamics \cite{vanetti2017piecewise}.
\subsubsection{Simulation}
Since the flow of the basic PDMP models is trivial to calculate, the computational burden of simulation is on the generation of the arrival times of the iPP(s) that determines the changes in velocity. Given that the velocity component is either Gaussian or uniform, the burden entirely depends on having tight bounds on the gradient. 
By the Markov property the process renews at each event, so we only need to consider the generation of the first event time $\tau_1$. For some arbitrary initial $(x_0,v_0)$, define $\varrho(t) = \int_0^{t} \rho(x_0 + s \cdot v_0)ds$, where $\rho$ is the rate function of the chosen PDMP. For an iPP, the probability of no arrivals in $[0, t]$ is then given by
\[
\mathbb P(\tau_1 > t) = \exp \Big \{ -\varrho(t)\Big \} = \exp \Big \{ -\int_0^{t} \rho(x_0 + s\cdot v_0) ds \Big \}.
\]
Sampling of $\tau_1$ can then be carried out by transforming the rate of the iPP to a homogeneous PP of rate 1, simulating an event-time of this process, and transforming the event-time back to the desired rate function via $\tau_1 = \varrho^{-1}(-\log u)$, where $u \sim U(0,1)$. Unfortunately, this iPP version of inversion sampling is only available in simple cases and instead, in the PDMP literature, the thinning method of \cite{lewis1979simulation} is used. To apply thinning, consider a fixed look-ahead $\theta > 0$. On $[0, \theta]$, we need the bound $\bar{\rho} = \max_{t \in [0, \theta]} \rho(x_t, v_t) + \gamma$, for some $\gamma > 0$. For log-concave functions the maximum is attained at the end point of the interval, $\bar{\rho}=\rho(x_t+\theta v_t,v_t)+\gamma$. The next event-time is then just simulated via $\tau \sim \text{Exp}(\bar{\rho})$, and accepted or rejected with probability $\rho(\tau)/\bar{\rho}$. Evaluation of the empirical rejection probability provides a direct measure of the efficiency of the bounding procedure. If $\tau > \theta$, instead update the look-ahead, calculate a new bound and generate a new event-time. Note that while $\theta$ might appear similar to a step-size, it is purely a computational parameter, and does not affect the mixing of the PDMP algorithms.
It is also worth noting that the trajectory length is directly correlated to the runtime, and as such does not contain any new information; yet it may provide a simple surrogate to the effective sample size and can also be helpful in planning the time needed for a simulation for a target Monte Carlo error.

\section{Examples\label{examples}}
In this section we compare the performance of the Langevin-based and the PDMP-based samplers in a number of numerical examples. They all illustrate that PDMPs can be used for \emph{exact} sampling when using non-differentiable priors. We also show that relaxing exactness (via MYE-based methods) does not necessarily lead to enhanced mixing compared to PDMP based approaches.

\subsection{Anisotropic Laplace\label{ex:100D_laplace}}
We consider here an anisotropic distribution inspired by \cite[Example 4.4]{bouchard2018bouncy} on $\mathcal X = \mathbb R^{100}$ given by 
$\pi(x) \propto \exp \left \{ - \beta^T \vert x \vert \right \},$ where $\beta = (1, 2, \ldots, 99, 100)^T$ is a vector of integers, $\vert\cdot\vert$ is applied element-wise, and superscript $T$ indicates transpose. These types of distributions generally prove challenging for MCMC samplers, as it is difficult to find a global stepsize that ensures efficient exploration across all dimensions. We approximate the potential $U(x)= \beta^T \vert x \vert$ by its Moreau-Yosida envelope with $\lambda=10^{-5}$, giving us our smoothed target $\pi^\lambda$ with a very tight envelope.  

We compare the performances of the Moreau-Yosida Unadjusted Langevin Algorithm (MY-ULA, \cite{durmus2018efficient}), our Moreau-Yosida Unadjusted Underdamped Langevin Algorithm (MY-UULA), SK-ROCK as introduced in \cite{vargas2019accelerating}, the proximal Metropolis-adjusted Langevin algorithm (pMALA, \cite{pereyra2016proximal}), the Bouncy Particle Sampler (BPS, \cite{bouchard2018bouncy}), and the Zig-Zag Sampler (ZZS, \cite{bierkens2019zig}). We let the algorithms run for an equivalent amount of wall clock time.

Figure \ref{fig:100D_laplace} shows four relevant marginal empirical densities for the six different samplers. We set the stepsize for MY-ULA to $\delta=\lambda/2$, the stepsize for SK-ROCK to $\delta=10^{-3}$, as in \cite{vargas2019accelerating}. For pMALA, we set $\delta=2\lambda_p$, where $\lambda_p=2\times10^{-5}$ is tuned to achieve an acceptance ratio of around $55\%$.
We observe that MY-ULA mixes very slowly in the most diffused marginal. Note that this is partly due to overdamped Langevin methods struggling with weakly log-concave distributions, see \cite{durmus2019high}, another challenge in this example is the anisotropy: while the algorithm mixes fast for the narrow marginals, it mixes slowly in the wide ones. Choosing a larger stepsize will result in faster mixing in these wide components, however it will also lead to the narrow marginals not being captured properly. 
We also note that the resulting behaviour is not a consequence of the tightness of the envelope: if, instead, $\lambda$ was of the order $10^{-1}$ or $10^{-2}$, corresponding to much less tight envelopes, MY-ULA and MY-UULA suffer similarly. There is in general no way to adjust for anisotropy with methods using a step size outside of preconditioning, which requires extensive knowledge of the target distribution in general.
Moving on to the asymptotically exact algorithms, pMALA shows slightly worse mixing behaviour compared to MY-ULA. 
Neither BPS nor ZZS degenerate in this scenario. The ZZS with its natural localization, provides very accurate estimates of the marginal densities across all dimensions. Due to the distribution factorizing into independent Laplace distributions, the BPS performs a bit worse than the ZZS in estimating marginals, as the latter makes use of the independence structure explicitly while the BPS is implemented naively in its global form. The BPS also requires roughly one refreshment in twenty reflections, or suffers from similar problems as the original BPS on Gaussian targets \cite{bouchard2018bouncy}. We emphasize here that a localized implementation of the BPS is possible in this example, and from our experience in this work gives better results than the global BPS. In Table \ref{table:ESS_Laplace} we see that SK-ROCK decorrelates very quickly for the narrow dimension (recall, however, that it fails to accurately capture the marginal), but for the most difficult wide dimension the ZZS outperforms SK-ROCK by a factor of four while targeting the correct invariant.

\begin{table}[htp]
\begin{center}
    \begin{tabular}{| l | c | c | c | c | c | c |}
    \hline
    Algorithm & MY-ULA & MY-UULA & SK-ROCK & pMALA & BPS & ZZS  \\ \hline\hline
    $\beta = 1$ & 2.0 & 2.3 & 6.0 & 1.7 & 3.0 & 24.9  \\ \hline
    $\beta = 100$ & 50.5 & 218.3 & 4197.4 & 182.9 & 755.5 & 2037.4  \\ 
    \hline
    \end{tabular}
\end{center}
\caption{Effective sample size per second for the fastest and slowest mixing dimensions of the anisotropic Laplace obtained from long runs of the respective algorithms. Recall that the first three algorithms are asymptotically biased, while the last three are asymptotically exact.}
\label{table:ESS_Laplace}
\end{table}

\begin{figure}\begin{center}
	\includegraphics[scale=0.7]{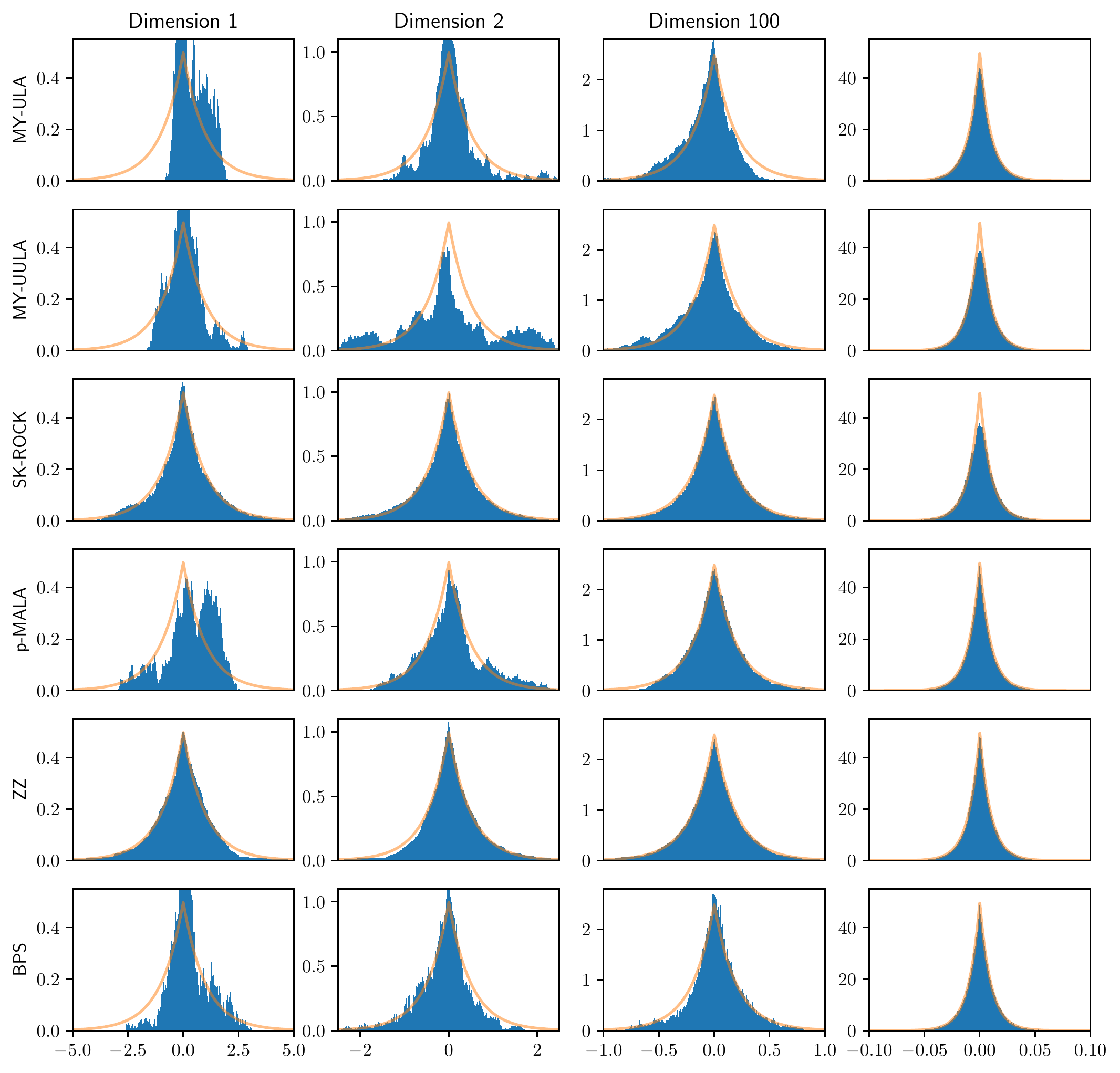}
	\caption{The first three rows correspond to the approximate algorithms, the last three are asymptotically exact. The first four are based on discretizations of the Langevin diffusions, the last two rows correspond to the ZZS and BPS samplers. Each column corresponds to histograms of the first, second, fifth, and hundredth dimension of the anisotropic Laplace, the underlying orange line shows the true marginals. All algorithms were given the same clock time for a fair comparison. After this time expired, some algorithms have not yet mixed well, explaining the difference between the true distribution and the histograms, despite the last three algorithms being asymptotically exact.}
	\label{fig:100D_laplace}
	\end{center}
\end{figure}

\subsection{Sparse Bayesian logistic regression\label{ex:sparsebay_numerics}}
Gene selection via sparse estimation has been subject of significant attention in the medical statistics literature, see for example \cite{shevade2003simple, cawley2006gene}, with a particular emphasis on the application of convex regularizers. The interpretation of frequentist regularizers in generalized linear models as negative log-priors has, however, had limited success: while the posterior mode coincides with the MLE, expected parameter values under a full posterior are significantly non-sparse due to the integral dependence on the posterior shape. To remedy this, scale-mixtures \cite{carvalho2010horseshoe, bhattacharya2015dirichlet, hosseini2019two} that combine significant mass near zero with heavy tails have seen widespread application. The resulting priors, however, are not log-concave, induce multimodality of the posterior, and often only satisfy a.e. differentiability. It is therefore not possible to apply MYE-based Langevin dynamics for this class of models as the Moreau-Yosida envelope is not well-defined due to the fact that the optimization problem \eqref{eq:mye_opt} might have multiple solutions. The resulting posteriors are typically sampled using custom Gibbs samplers, we here show the general applicability of PDMPs to target such distributions. 
In particular, we consider a sparse Bayesian logistic regression model with a Bessel-K prior, see \cite{hosseini2019two} for details on this prior specification. We let $y \in \mathbb \{-1,1\}^d$ be a binary vector, $Z \in \mathbb R^n \times \mathbb R^d$ a matrix of covariates, and $x \in \mathbb R^n$ a parameter vector. For each observation $i$ we model the outcome as a Bernoulli trial with probability of success given by $\text{logistic}\left(Z_i^T x\right)$ where $Z_i$ is the $i$'th column and $\text{logistic}(z) = 1/(1+\exp(-z))$. The Bessel-K prior with parameters $(p, \epsilon)$ is given by $p(x) \propto \vert x + \epsilon \vert^{p-\frac{1}{2}}K_{p-\frac{1}{2}}\left (\vert x \vert + \epsilon \right )$, where $K_\alpha$ is the modified Bessel function of the second kind with order $\alpha$, an illustration is given in Figure \ref{fig:sparse_genes}, top-left. As the factor graph is dense, the BPS and variations are not a suitable choice in this case. The ZZS, however, is more amenable to models that feature a fully connected factor graph. We will use constant bounds, which, while overly conservative, alleviate the dominant computational cost of running the ZZS, calculating local bounds, with an operation of order $O(1)$.  

We will use a dataset derived from \cite{golub1999molecular}, which consists of $n=7129$ gene expressions with associated parameters $x_j, j =1,2,\ldots, 7129,$ of $d=72$ individuals with either acute myeloid or acute lymphoblastic leukaemia. Rather than classification of tumour type, our aim is to discover expressions that significantly contribute to the classification and warrant further scrutiny, and we therefore pool the test and training set. Based on the discussion in \cite{hosseini2019two}, we let $p = 0.002$ and $\epsilon = 0.05$, generated $10^5$ samples and discarded the first $10^4$. Trace plots of parameters and the log-energy indicate the sampler mixed well. \cite{bhattacharya2015dirichlet} applies a sequential 2-means cluster procedure to estimate the number of sparse parameters for a similar problem, we instead post-hoc subset some percentile of the largest absolute posterior means of $x$. If we subset the 0.2\% percent largest values, the 16 resulting expressions results in a perfect prediction on the pooled set, see top right in Figure \ref{fig:sparse_genes}, but include none of the ones found in \cite{cawley2006gene}; the second percentile of the largest absolute components of $x$, 144 in total, include 10 of the 11 gene expressions found by the algorithm of \cite{cawley2006gene} via leave-one-out cross-validation. This indicates that the model successfully recovers novel and relevant genetic expressions not found via standard coordinate descent methods. The low number of observations implies that the effect size is limited, subsequently the 90\% credible intervals all include zero; if the data set was larger we would expect the credible intervals to concentrate more strongly around the contributing genes expressions. We plot the posterior median and means with 90\% credible intervals in Figure \ref{fig:sparse_genes}, bottom. In experiments, not reported here, using a Laplace prior where the number of observations $n$ exceeds the number of parameters $d$, the subsampling property lead to significant outperformance for the ZZS relative to the Langevin-based schemes.
\begin{figure}[ht!]
\begin{center}
    \centering
    \includegraphics[scale=0.5]{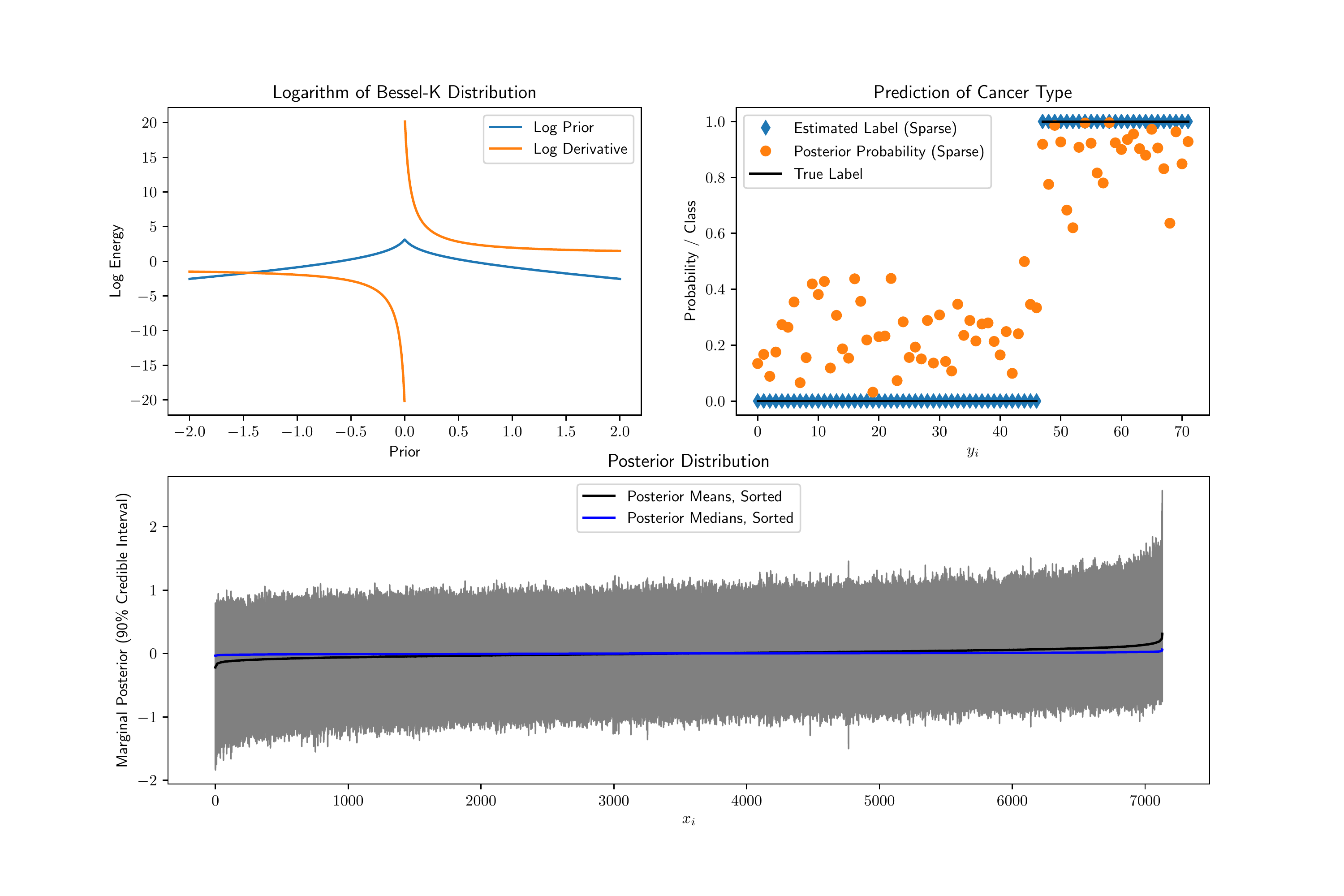}
    \caption{Top left: Logarithm of a Bessel-K prior for a univariate case with $(p,\epsilon) = (0.002, 0.05)$ and corresponding log-derivative. Top right: Prediction probabilities for the pooled data set using only 16 genes. Bottom: Posterior mean, median and 90\% credible intervals. }
    \label{fig:sparse_genes}
\end{center}
\end{figure}

\subsection{Circular Bayesian Statistics\label{ex:ants}}
Circular statistics addresses inference from periodic data such as angles and rotations, for example when analyzing the direction ants move in when responding to an evenly illuminated black target \cite{jander1957optische}, this can be modelled by an asymmetric wrapped Laplace distribution \cite{fernandez2004circular}. This example illustrates the applicability of PDMPs when dealing with circular, non-differentiable, and multi-modal posteriors. Here again, the MYE is not well-defined, precluding the use of the samplers based on the MYE-approximation.

In \cite{jander1957optische}, the author studies which direction ants walk in when being placed in the middle of an evenly illuminated arena with two black discs on the side. These directions were observed for 253 ants, the observations are summarised in Figure \ref{fig:ants}. For the one-disc example, \cite{fernandez2004circular} shows that the wrapped Laplace distribution is a good fit to the observations. For the two-disc model, we thus assume a mixture of two wrapped Laplace distributions with means $\mu_i\in[0,2\pi)$, scale parameters $\lambda_i>0$, and skew parameters $\kappa_i>0$, $i\in\{1,2\}$. The likelihood of a data point $y\in[0,2\pi)$ given the mixture component is calculated by
\begin{align*}
    \theta &= \Bigg\{\begin{array}{lr}
        y-\mu_i, & \text{for } y-\mu_i>0\\
        y-\mu_i+2\pi, & \text{for } y-\mu_i\leq0
        \end{array}\\
    L(\theta|\mu_i,\lambda_i,\kappa_i)&=\frac{\lambda_i\kappa_i}{1+\kappa_i^2}\left(\frac{e^{-\lambda_i\kappa_i\theta}}{1-e^{-2\pi\lambda_i\kappa_i}}+\frac{ e^{(\lambda_i/\kappa_i)\theta}}{e^{2\pi(\lambda_i/\kappa_i)}-1}\right),
\end{align*}
where the definition of the auxiliary variable $\theta$ handles the periodic extension due to the shift by the mean. The prior on the $\mu_i$ are uniform distributions on $[0,2\pi]$, the ones on the scale parameters $\lambda_i$ are Exponential$(1)$ distributions, the one on the $\kappa_i$ are Gamma$(2,1/2)$ distributions, and the one on the mixture parameter $\rho$ is a Beta$(100,100)$ distribution. Figure \ref{fig:ants} shows the posterior distribution for $\mu_1$ as estimated by the BPS. The reader should in particular note the multi-modality of the distribution, which especially the BPS samples from effectively: Table \ref{table:ESS_ants} summarises the effective sample size per second for the BPS, the ZZS, and a Random Walk Metropolis Hastings Sampler. \\

\begin{minipage}{\textwidth}
  \begin{minipage}[b]{0.49\textwidth}
    \centering
    \includegraphics[scale=0.16]{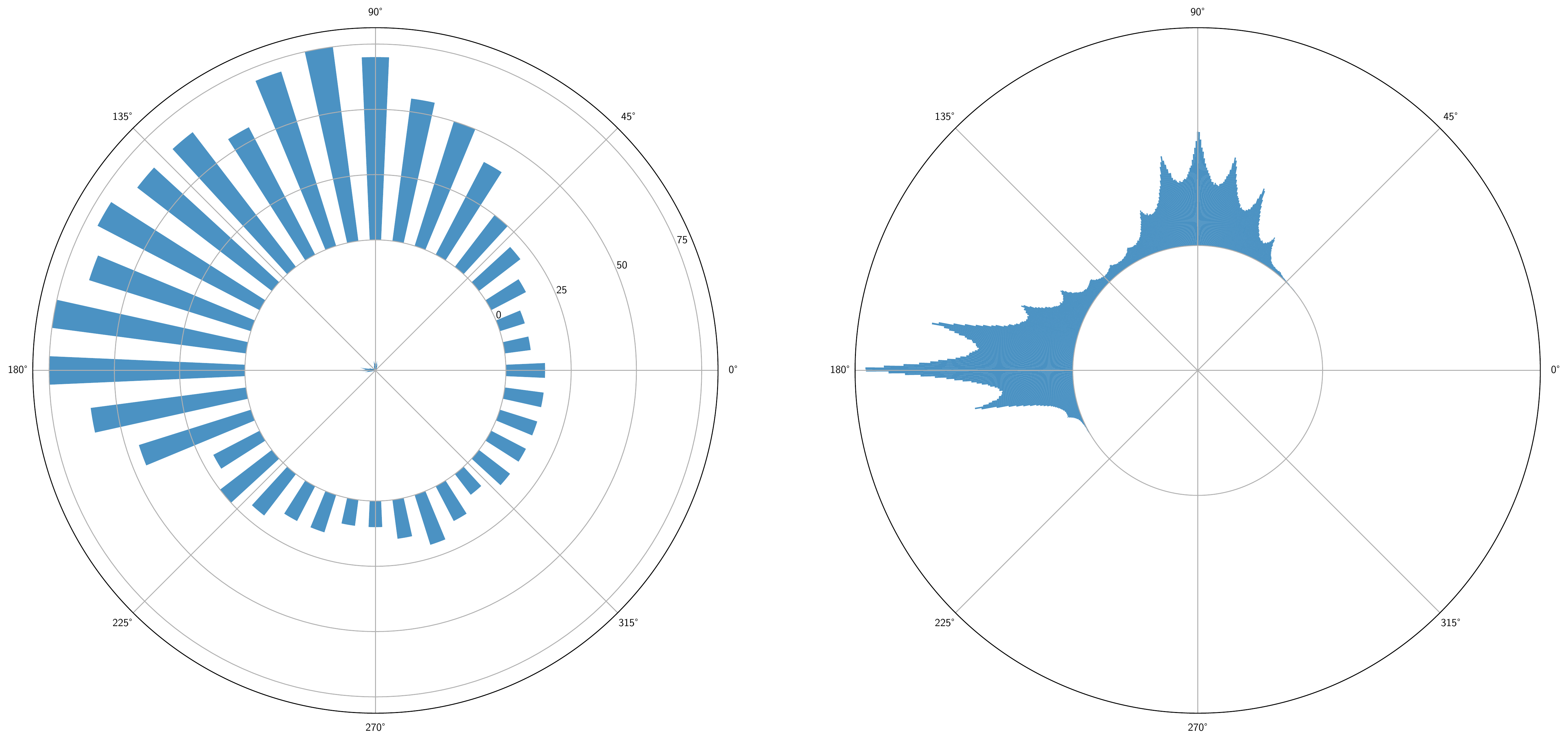}
    \label{fig:ants}
    \captionof{figure}{Left: Observations, taken from \cite[Fig. 18B]{jander1957optische}. Right: Marginal posterior density for the mean parameter $\mu_1$.}
  \end{minipage}
  \hfill
  \begin{minipage}[b]{0.49\textwidth}
    \centering
    \begin{tabular}{| l | c | c | c | c |}
    \hline
    Algorithm & $\mu_1$ & $\lambda_1$ & $\kappa_1$ & $\rho$ \\ \hline\hline
    BPS & 3.36 & 164.80 & 6.29 & 2095.44 \\ \hline
    ZZS & 0.66 & 39.53 & 1.12 & 537.33 \\ \hline
    RWMH & 2.88 & 37.01 & 4.52 & 1153.13 \\ \hline
    \end{tabular}    \label{table:ESS_ants}
     \captionof{table}{ESS/s for different variables. The ESS/s for the variables from the second mixture are similar, as is expected due to the mixture components being indistinguishable from one another.}
    \end{minipage}
\end{minipage}

\section{Discussion\label{section:discussion}}
Both PDMPs and diffusion-based samplers have been shown to work in various examples, albeit differently well. This discussion aims to provide the reader with an understanding of the strengths and weaknesses of the respective methods, and to guide the practitioner as to which algorithm to use. In the following, we discuss the key aspects one needs to consider.

\emph{Dimensionality:} PDMPs perform worse as the dimension increases: The ZZS requires at least $d$ events to completely change the direction of the velocity vector, and for each of these event the gradient in the respective direction needs to be re-evaluated, which (depending on the problem) can become prohibitively expensive. The BPS is known to suffer in high-dimensions too, with problems arising even when targeting isotropic Gaussians \cite{bouchard2018bouncy} requiring many refreshments. If the problem is localizable (such as in Example \ref{ex:image_deblurring}), the local version of the BPS can alleviate these issues by reducing the problem size to multiple smaller problems. An interesting direction for future research would be to parallelize the dynamics of these smaller problems, which would be a strong argument for PDMPs in high-dimensional settings, as discussed in Example \ref{ex:image_deblurring}. The smoothing by the MYE always results in a non-localizable target, such that the computation of the proximal operator can not be broken down into smaller problems. Furthermore, in the large data regime, data subsampling is straightforward for the PDMP samplers.

\emph{Anisotropic targets:} As illustrated in Example \ref{ex:100D_laplace}, and similarly in (the everywhere differentiable) Example \ref{ex:100D_gaussian} in the appendix, especially the ZZS is able to adapt to highly anisotropic targets. The diffusion-based samplers would improve if one has a preconditioner available, but if this is not the case, they struggle: A tight envelope results in small step sizes such that mixing in the `slowly mixing components' takes very long. However, if one chooses a rather crude approximation, the approximation error in the `fast mixing components' grows. This is graphically visible in Figures \ref{fig:100D_laplace} and \ref{fig:100D_gaussian}. 

\emph{Log-concavity:} The MYE is only well-defined for log-concave targets, and as illustrated in Examples \ref{ex:sparsebay_numerics} and \ref{ex:ants}, PDMPs allow targeting these non-differentiable posteriors, while the diffusion-based algorithms rely on the MYE to be well-defined. In weakly log-concave settings such as Example \ref{ex:100D_laplace} when the gradient of the target is not Lipschitz continuous, a very tight envelope is needed to ensure a good approximation: This however, results in small step sizes and thus slow mixing. Furthermore, the samplers based on the overdamped Langevin diffusion struggle in this setting as the gradients do not grow when $|x|$ diverges. In strongly log-concave targets, such as Examples \ref{ex:checkerboard} and \ref{ex:image_deblurring}, the tuning guidance in \cite{vargas2019accelerating} and \cite{durmus2018efficient} ensures reasonable approximations while at the same time facilitating fast mixing. In summary, for strongly log-concave targets the diffusion-based algorithms are often preferable, while PDMPs are  a viable option when the MYE is not well-defined.

\emph{Exact Sampling:} The PDMPs are inherently asymptotically exact samplers, while the unadjusted diffusion-based samplers are not. However, it is possible to incorporate a Metropolis-Hastings correction step in MY-ULA to account for this, giving the proximal Metropolis-Adjusted Langevin Algorithm (pMALA, \cite{pereyra2016proximal}). pMALA is recommended to be tuned such that one achieves around $50\%$ - $70\%$ acceptances, which may result in slower mixing in comparison to the unadjusted algorithms. 

\emph{Proximal Operators and Event Rates:} The proximal operators for some functions are available in closed form (see e.g. \cite{polson2015proximal}), however often the minimization problems needs to be solved numerically, in which case the evaluation becomes expensive. On the other end, the PDMPs will perform significantly worse if one can't find good bound on the event rates, as in that case one needs to evaluate the gradients more often.

\section{Conclusion\label{conclusion}}
We have in this work shown that sampling algorithms based on PDMPs allow exact sampling from a range of non-differentiable target distributions. In particular, we exploit gradients if they exist almost everywhere, which is a common scenario in many real-world applications using non-differentiable priors. This has particular relevance in cases like cancer tumor classification \cite{golub1999molecular} where accuracy is at a premium, and in situations where the proximal operator is prohibitively expensive to calculate. In comparison to gradient-based methods, PDMPs can naturally handle anisotropy of the posterior without preconditioning, and furthermore be localized whenever the distribution is of product-form. In contrast, the unadjusted diffusion algorithms perform well in cases where low accuracy of the envelope is acceptable and when the posterior exhibits strong log-concavity. Utilizing second-order information in the MY-UULA algorithm can prove beneficial to mixing. Furthermore, we have illustrated efficient sampling of non-convex posteriors where Moreau-Yosida methods are not applicable. In conclusion, we have shown that PDMPs are a very able tool whenever accurate posterior inference is required in complex non-differentiable scenarios, and discussed which distribution characteristics suggest one or the other class of algorithms to be preferable.

\printbibliography

\newpage
\appendix

\section{Supplementary Material - Proofs\label{appendixA}}
\subsection{Proof of Theorem \ref{thm_expectation_bounds}\label{proof:thm_expectation_bounds}}

\begin{proof}[Proof of Theorem \ref{thm_expectation_bounds}]
From Theorem \ref{thm_MYE_closeness} we immediately get the inequalities
\begin{align}
    -g(x)&\leq-g^\lambda(x)\label{ineq1}\\
    -g^\lambda(x)&\leq-g(x)+\frac{L^2\lambda}{2}\label{ineq2},
\end{align}
and thus also
\begin{align}
    \int\exp(-g^\lambda(z))dz&\geq\int\exp(-g(z))dz=1\label{ineq3}
\end{align}
and
\begin{align}
    \int\exp(-g^\lambda(z))dz&\leq\int\exp(-g(z))\exp(L^2\lambda/2)dz=\exp(L^2\lambda/2)\label{ineq5}.
\end{align}

Let $f\geq0$, then
\begin{align*}
    \mathbb E_{\pi^\lambda}(f)&= \int f(x)\frac{\exp(-g^\lambda(x))}{\int\exp(-g^\lambda(z))dz}dx\\
    &\overset{(\ref{ineq3})}{\leq} \int f(x)\exp(-g^\lambda(x))dx\\
    &\overset{(\ref{ineq2})}{\leq} \exp(L^2\lambda/2)\int f(x)\exp(-g(x))dx =\exp(L^2\lambda/2)\mathbb E_\pi(f).
\end{align*}
Similarly, again for $f\geq0$,
\begin{align*}
    \mathbb E_{\pi^\lambda}(f)&= \int f(x)\frac{\exp(-g^\lambda(x))}{\int\exp(-g^\lambda(z))dz}dx\\
    &\overset{(\ref{ineq5})}{\geq} \exp(-L^2\lambda/2)\int f(x)\exp(-g^\lambda(x))dx\\
    &\overset{(\ref{ineq1})}{\geq} \exp(-L^2\lambda/2)\int f(x)\exp(-g(x))dx\\
    &= \exp(-L^2\lambda/2)\mathbb E_\pi(f).
\end{align*}
In summary, for any non-negative $f$, we have
\begin{align}\label{important_ineq}
    \exp(-L^2\lambda/2)\mathbb E_\pi(f)\leq\mathbb E_{\pi^\lambda}(f)\leq\exp(L^2\lambda/2)\mathbb E_\pi(f).
\end{align}

Subtracting $\mathbb E_\pi(f)\geq0$ from these inequalities lets us derive
\begin{equation}\label{proof_ineq:twobounds}\begin{split}
    -(\exp(L^2\lambda/2)-1)\mathbb E_\pi(f)
    &=-\max\{\exp(L^2\lambda/2)-1,1-\exp(-L^2\lambda/2)\}\mathbb E_\pi(f)\\
    &=\min\{1-\exp(L^2\lambda/2),\exp(-L^2\lambda/2)-1\}\mathbb E_\pi(f)\\
    &\leq(\exp(-L^2\lambda/2)-1)\mathbb E_\pi(f)\\
    &\overset{(\ref{important_ineq})}{\leq}\mathbb E_{\pi^\lambda}(f)-\mathbb E_\pi(f)\\
    &\overset{(\ref{important_ineq})}{\leq}(\exp(L^2\lambda/2)-1)\mathbb E_\pi(f)\\
    &\leq\max\{\exp(L^2\lambda/2)-1,1-\exp(-L^2\lambda/2)\}\mathbb E_\pi(f)\\
    &=(\exp(L^2\lambda/2)-1)\mathbb E_\pi(f),
\end{split}
\end{equation}
and therefore
\begin{align}\label{important_ineq2}
    |\mathbb E_{\pi^\lambda}(f)-\mathbb E_\pi(f)|\leq(\exp(L^2\lambda)-1)\mathbb E_\pi(f)
\end{align}
holds for any non-negative $f$.

For general $f$, we consider the standard decomposition $f=f^+-f^-$ with $f^+\geq0$ and $f^-\geq0$. Then $|f|=f^++f^-$, and as
\begin{align*}
    |\mathbb E_{\pi^\lambda}(f)-\mathbb E_\pi(f)|&=|\mathbb E_{\pi^\lambda}(f^+)-\mathbb E_\pi(f^+)-[\mathbb E_{\pi^\lambda}(f^-)-\mathbb E_\pi(f^-)]|\\
    &\leq |\mathbb E_{\pi^\lambda}(f^+)-\mathbb E_\pi(f^+)|+|\mathbb E_{\pi^\lambda}(f^-)-\mathbb E_\pi(f^-)|\\
    &\overset{(\ref{important_ineq2})}{\leq}(\exp(L^2\lambda)-1)\mathbb E_\pi(f^+)+(\exp(L^2\lambda)-1)\mathbb E_\pi(f^-)\\
    &=(\exp(L^2\lambda)-1)\mathbb E_\pi(|f|),
\end{align*}
we have proved the first bound in the Theorem.

Since we can exchange the roles of $\pi$ and $\pi_\lambda$ in (\ref{important_ineq}), we can follow the same chain of arguments to also get
\begin{align*}
|\mathbb E_{\pi^\lambda}(f)-\mathbb E_\pi(f)|\leq(\exp(L^2\lambda/2)-1)\mathbb E_{\pi^\lambda}(|f|).
\end{align*}
If $g=g_1+g_2$ with Lipschitz-continuous $g_1$ and differentiable, but not necessarily Lipschitz-continuous, $g_2$, one takes the MYE of $g_1$ and notes that \ref{ineq1} and \ref{ineq2} hold for $g_1$. Adding $g_2$ on both sides of the inequality shows that these inequalities remain true for $g$ such that the proof still holds.
\end{proof}

\subsection{Proof of Lemma \ref{thm_gradient_bounds}\label{proof:thm_gradient_bounds}}

\begin{proof}
The case $\lambda_1=\lambda_2$ is trivial so assume $\lambda_1<\lambda_2$.\\
Firstly recall that for convex $g$ any MYE is also convex. Further note that $g^{\lambda_2}$ is a Moreau-Yosida envelope for $g^{\lambda_1}$, with $g^{\lambda_2}=(g^{\lambda_1})^{\lambda_2-\lambda_1}$ \cite[Proposition 12.22 (ii)]{bauschke2011convex}.\\
We may thus define $h=g^{\lambda_1}$, $\lambda=\lambda_2-\lambda_1$, such that the statement of the lemma is equivalent to
\begin{lemma*}[Equivalent Formulation of Lemma \ref{thm_gradient_bounds}]
For any convex and differentiable function $h:\mathcal X\rightarrow ]-\infty,\infty]$, and for any $\lambda>0$, the Moreau-Yosida envelope $h^\lambda$ satisfies
\begin{align*}
    \lVert \nabla h(x)\rVert\geq\lVert\nabla h^{\lambda}(x)\rVert\quad\forall x\in\mathcal X.
\end{align*}
\end{lemma*}
We define $p=\text{prox}_h^\lambda(x)$. By theorem \ref{theorem:log_gradient_prox}, $\nabla h^\lambda(x)=(x-p)/\lambda$; and by convexity (and differentiability) of $h$, we have for any $x\in\mathcal X$:
\begin{align*}
    0&\leq\langle\nabla h(p)-\nabla h(x),p-x\rangle\\
    &=\langle\nabla h(p)-\nabla h(x),-\lambda\nabla h(p)\rangle\\
    &=-\lambda\lVert\nabla h(p)\rVert^2+\langle\nabla h(x),\nabla h(p)\rangle\\
    &\leq-\lambda\lVert\nabla h(p)\rVert^2+\frac\lambda2\lVert\nabla h(x)\rVert^2+\frac\lambda2\lVert\nabla h(p)\rVert^2\\
    &=\frac\lambda2\lVert\nabla h(x)\rVert^2-\frac\lambda2\lVert\nabla h(p)\rVert^2,
\end{align*}
where the first inequality is a necessary and sufficient condition for convexity of a differentiable function, and the last inequality follows from Young's inequality as $\langle x,y\rangle\leq\lVert x\rVert^2/2+\lVert y\rVert^2/2$. 
\begin{align*}
    \lVert\nabla h(x)\rVert^2&\geq\lVert\nabla h(p)\rVert^2 \overset{(\ref{prox_op_1})}{=}\lVert\frac1\lambda(x-p)\rVert =\lVert\nabla h^{\lambda}(x)\rVert
\end{align*}
as required. The last equality is given by Theorem \ref{theorem:log_gradient_prox}.
\end{proof}

\subsection{Proof of Lemma \ref{lemma:zz_invariance}\label{proof:zz_invariance}}
\begin{proof}
Invariance follows if 
\[
\int \mathcal L_{ZZ} f(x,v) \pi(dx) p(dv) = \int_{A_0} \mathcal L_{ZZ} f(x,v) \pi(dx) p(dv) + \int_{A_0^c} \mathcal L_{ZZ} f(x,v) \pi(dx) p(dv) = 0 
\]
for any $f \in D(\mathcal L_{ZZ})$, the domain of $\mathcal L_{ZZ}$ \cite{ethier2009markov}. If the prior is differentiable such that $\pi$ has a differentiable density, $A_0$ is empty and the proof is directly as in \cite[Theorem 2.2]{bierkens2019zig}. If the prior is non-differentiable, $A_0$ is non-empty but a null-set under $n$-dimensional Lebesgue measure. Since $\pi$ and $p(dv)$ are absolutely continuous with respect to Lebesgue measure, it follows that the first integral is zero. 
Invariance then again follows directly from \cite[Theorem 2.2]{bierkens2019zig}.
\end{proof}

\subsection{Proof of Lemma \ref{lemma:bps_invariance}\label{proof:bps_invariance}}
\begin{proof}
As in Lemma \ref{lemma:zz_invariance}, invariance with respect to the joint distribution of $(x,v)$ follows if 
\[
\int \mathcal L_{BPS} f(x,v) \pi(dx) p(dv) = \int_{A_0} \mathcal L_{BPS} f(x,v) \pi(dx) p(dv) + \int_{A_0^c} \mathcal L_{BPS} f(x,v) \pi(dx) p(dv) = 0
\]
for any $f \in D(\mathcal L_{BPS})$. Similarly to the proof of Lemma \ref{lemma:zz_invariance}, the proof of \cite[Proposition 1]{bouchard2018bouncy} applies directly under absolute continuity of $\pi$ and $p(v)$ with respect to Lebesgue measure.
\end{proof}

\subsection{Discretizing the Underdamped Langevin Dynamics\label{appendix:discretisation}}
We implement the discretization used in \cite{ma2019there}. If the current position and velocity are $(x_t,v_t)$, the next iteration is given by
\begin{align*}
    \begin{cases}
x_{t+1}=x_t+\frac{1-\beta}{\gamma}v_t-\frac{1}{\gamma}(\nu-\frac{1-\beta}{\gamma\xi})\nabla U^\lambda(x_t)+W_x\\
v_{t+1}=\beta v_t-\frac{1-\beta}{\gamma\xi}\nabla U^\lambda(x_t)+W_v,
\end{cases}
\end{align*}
where $\nu=t_{n+1}-t_n$ is the step size, $\beta=\exp(-\gamma\xi\nu)$, and $(W_x,W_v)\sim\mathcal N(0,\Sigma)$ is Gaussian noise with covariance
\begin{align*}
    \Sigma=\begin{pmatrix} 
\frac1\gamma\left(2\nu-\frac{3}{\gamma\xi}+\frac{4\beta}{\gamma\xi}-\frac{\beta^2}{\gamma\xi}\right)I_{d\times d} & \frac{1+\beta^2-2\beta}{\gamma\xi}I_{d\times d} \\
\frac{1+\beta^2-2\beta}{\gamma\xi}I_{d\times d} & \frac{1-\beta^2}{\xi}I_{d\times d}
\end{pmatrix}.
\end{align*}
All the experiments in this paper were run with $\gamma=2$, $L=1\lambda$, and $\nu=2\lambda$, where $\lambda$ is the tightness parameter of the respective MYE.

\subsection{Hamiltonian Bouncy Particle Sampler}\label{supp:hamiltonian}
An alternative specification for the dynamics of the BPS was introduced in \cite{vanetti2017piecewise}, which we will now detail. Consider the Hamiltonian of both the target variable and the velocity $H(x, v)$
\begin{align*}
    H(x,v) = U(x) + \log p(v) = -\ell(x) - \log \pi_0 (x) - \frac{1}{2} v^t v + c,
\end{align*}
where $c$ is some constant we will suppress from now on. For some spherical potential $V(x) = \frac{1}{2} (x-\mu)^t \Sigma^{-1} (x-\mu)$, consider now the augmented Hamiltonian
\begin{align*}
    H(x, v) = \underbrace{-\ell(x)  - \log \pi_0 - V(x)}_{\hat U(x)} + \underbrace{V(x) - \frac{1}{2} v^t}_{\hat{H}(x,v)},
\end{align*}
which naturally can be broken into two parts. For the new system consisting of the latter two terms, the dynamics of the Hamiltonian $\hat H$ are available in closed form since the system of ODEs 
\begin{align*}
    \frac{\partial v_t}{\partial t} &= -\nabla_{x_t} \hat{H}(x_t, v_t) = -\Sigma^{-1}(x_t-\mu) \\
    \frac{\partial x_t}{\partial t} &=  \nabla_{v_t} \hat{H}(x_t, v_t) = v_t
\end{align*}
can be solved explicitly for any $\mu$ and $\Sigma$. For the first three terms, we note that if the model under consideration has a Gaussian component in $x$ then the spherical potential can be chosen to equal this energy function. For example, if $\pi_0$ is Gaussian, setting $V(x) = -\log \pi_0(x)$ reduces the Hamiltonian to only depend on the likelihood. As shown in \cite{vanetti2017piecewise}, the resulting Hamiltonian BPS with rates and reflection operator given by
\begin{align*}
    \hat{\rho}(t) &= \max \{0, \langle v_t, \nabla \hat U(x_t) \rangle \} \\
    \hat{\mathfrak{R}}_x v &= v - 2 \frac{\langle v, \nabla \hat U(x) \rangle}{\Vert \nabla \hat U(x) \Vert^2 }\nabla \hat U(x),
\end{align*}
and flow determined by $\hat H(x,v)$, has $\pi(x) v(x)$ as invariant distribution. It is clear that this is valid under any choice of $\mu$ and $\Sigma$, in particular, the Hamiltonian BPS can be localized if a factor decomposition is explicitly available. 
\section{Supplementary Material - Further Experiments\label{appendixB}}

\subsection{Anisotropic Gaussian\label{ex:100D_gaussian}}
To assess how the different algorithms compare on a \emph{strongly} log-concave example, we repeated Example \ref{ex:100D_laplace} with a centered Gaussian distribution, as in \cite[Example 4.4]{bouchard2018bouncy}. The $100$-dimensional distribution has a diagonal covariance matrix with $\Sigma_{i,i}=1/i^2$. Following the guidance in \cite{durmus2018efficient}, we picked $\lambda=1/10^4$, as this is the Lipschitz constant of the log-gradient, and chose a step size $\delta/2=\lambda$ for MY-ULA, and $\delta=0.005$ for SK-ROCK. For pMALA, we set $\delta=2\lambda$, and then chose $\lambda=3\times10^{-5}$, giving an acceptance probability of around $60\%$. The results are summarised in Figure \ref{fig:100D_gaussian}. The BPS is again run in its global form, a localized version thereof would improve performance. Estimates of the effective sample size per second are summarised in Table \ref{table:ESS_Gaussian}.

\begin{table}[htp]
\begin{center}
    \begin{tabular}{| l | c | c | c | c | c | c |}
    \hline
    Algorithm & MY-ULA & MY-UULA & SK-ROCK & pMALA & BPS & ZZS  \\ \hline\hline
    $\beta = 1$ & 2.00 & 1.39 & 2.28 & 1.27 & 1.17 & 2.48  \\ \hline
    $\beta = 100$ & 1774.74 & 1257.58 & 195.86 & 551.61 & 809.93 & 312.29  \\ 
    \hline
    \end{tabular}
\end{center}
\caption{Effective sample size per second for the different algorithm when targeting an anisotropic Gaussian distribution. Recall that the first three algorithms are asymptotically biased, while the last three are asymptotically exact.}
\label{table:ESS_Gaussian}
\end{table}

\begin{figure}\begin{center}
	\includegraphics[scale=0.7]{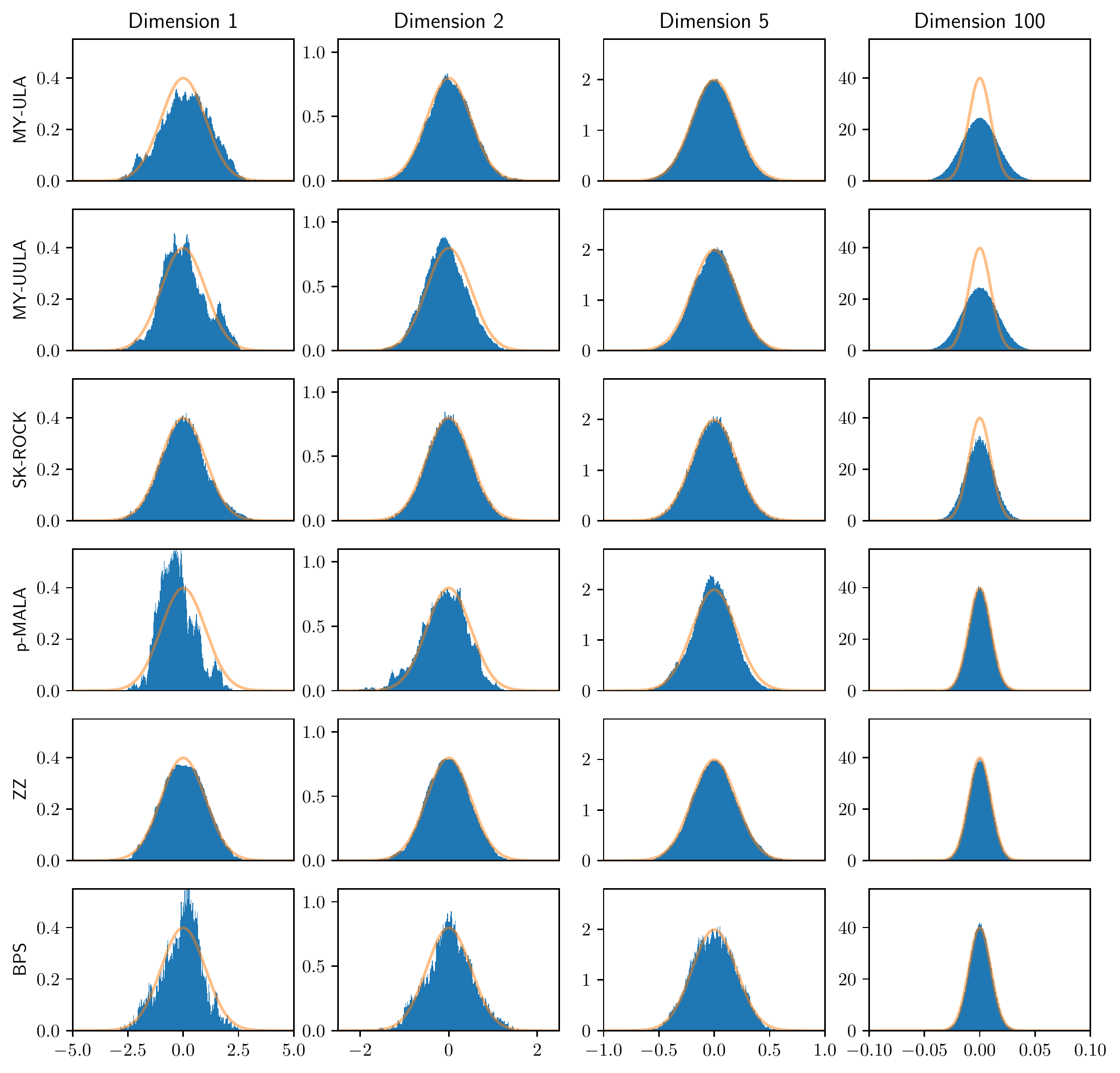}
	\caption{All algorithms are targeting a $100$-dimensional anisotropic Gaussian distribution. The first three rows correspond to the approximate algorithms, and none of them manage to fully capture the narrowest component. The ZZS perfectly captures the last component, and shows good results in the first component. The BPS (in its global form) mixes slowly in the first component, but well in the last. All algorithms were given the same computational budget for a fair comparison.}
	\label{fig:100D_gaussian}
	\end{center}
\end{figure}

\subsection{Nuclear-norm models for low-rank matrix estimation\label{ex:checkerboard}}
As a final illustration of our methods performance in exact sampling, we consider a nuclear-norm model example taken from \cite{pereyra2016proximal}. Let $x\in\mathbb R^{n\times n}$ be an unknown low-rank matrix, and let observations be noisy measurements thereof: $y=x+\xi$, where the entries of $\xi$ are i.i.d. $N(0,\sigma^2)$. We assume that $x$ is a low-rank matrix, and our aim is to sample from the posterior distribution of $x$ given by
\begin{align}
    \pi(x)\propto\exp\left(-\frac{1}{2\sigma^2}\lVert x-y\rVert_F-\alpha\lVert x\rVert_*\right),
\end{align}
where $\lVert\cdot\rVert_F$ denotes the Frobenius norm and $\lVert\cdot\rVert_*$ denotes the nuclear norm which favors low-rank matrices and penalizes high-rank ones. Conveniently, the proximal operator of the nuclear norm is available in closed form: Let $x=Q\Sigma V^T$ be the singular value decomposition of $x$, with $\Sigma=\text{diag}(\sigma_1,\dots,\sigma_n)$. Then the proximal operator is given by
\begin{align*}
    \text{prox}^\lambda_{\alpha\lVert\cdot\rVert_*}(x)=Q\text{diag}\left(\text{sgn}(\sigma_1)\max(|\sigma_1|-\alpha\lambda, 0),\dots,\text{sgn}(\sigma_n)\max(|\sigma_n|-\alpha\lambda,0)\right)V^T,
\end{align*}
i.e., one applies the soft thresholding operator to the singular values of $x$. We can thus efficiently compute the gradient to use in the Langevin-based samplers,
\begin{align}
    \nabla U^\lambda(x) = \frac{1}{\sigma^2}(x-y)+\frac{1}{\lambda} \Big (x - \text{prox}_{\alpha\lVert\cdot\rVert_*}^\lambda(x) \Big ).
\end{align}
We generated $y$ by adding Gaussian noise to a matrix $x^{\text{true}}$ with entries $x_{i,j}^{\text{true}}\in \{0,0.7,1\}$. The matrix $x^{\text{true}}$ is visually a checkerboard with white, grey, or black checks.

We set $\lambda=\sigma^2$. The step size for MY-ULA is set to $\delta=2\lambda$. A particular issue for the BPS in this model is the lack of factor decomposition due both to non-linearity of the nuclear norm and the proximal operator, which prevents us from using a localized, and therefore faster, version of the BPS. In an attempt to mitigate the resulting debilitated dynamics, we note that the likelihood in this case is equivalent to a isotropic Gaussian distribution in $x$ as well. Defining an auxiliary potential by $V(x\vert y) = \Vert x-y \Vert^2$/2, we propose to generate dynamics according to the Hamiltonian flow (see \ref{supp:hamiltonian}) corresponding to $(\dot{x}, \dot v) = (v_t, -(x_t-y)/{\sigma^2})$, which has the explicit solution
\begin{align*}
    \begin{pmatrix}
    x_t \\
    v_t 
    \end{pmatrix} = 
    \begin{pmatrix}
    v_0 \sin\left (\frac{t}{\sigma} \right ) \sigma + (x_0-y) \cos \left ( \frac{t}{\sigma} \right ) + y \\
    -(x_0-y) \sin\left (\frac{t}{\sigma} \right ) + v_0 \cos \left ( \frac{t}{\sigma} \right )
    \end{pmatrix}.
\end{align*}
By this choice of $V$ it follows that the gradient employed in the rate and reflection operator subsequently is
\begin{align}
    \nabla\hat U^\lambda(x)=\frac{1}{\lambda}\Big(x-\text{prox}_{\alpha\lVert\cdot\rVert_*}^\lambda(x)\Big).
\end{align}

Figure \ref{fig:CB_results} shows the mean squared error between the posterior mean estimate of the respective algorithms, as calculated every second, and the `true' posterior mean, as estimated by a very long run using an asymptotically unbiased algorithm. All algorithms are started at the same point, not too far away from the region of high probability. One can see that while MY-ULA quickly gives good estimates, the second-order scheme MY-UULA quickly yields better estimates. Interestingly, SK-ROCK performs worse here. The BPS does not yield any useful estimates in reasonable time, but after a while the HBPS gives the second best results. For completeness, we note that the Zig-Zag Sampler is not able to computationally compete with any of the other methods, as a single reflection requires the evaluation of the full gradient, which is prohibitively expensive. We also estimated the slowest and fastest mixing components of the checkerboard by estimating the sample covariance matrix during a long run of an exact sampler, and taking the first and last eigenvector thereof as the direction where the chain mixes slowest, and fastest, respectively. The autocorrelation plots for these components are shown in the second and third panel of Figure \ref{fig:CB_results}.

\begin{minipage}{\textwidth}
  \begin{minipage}[b]{0.3\textwidth}
    \centering
    \includegraphics[scale=0.4]{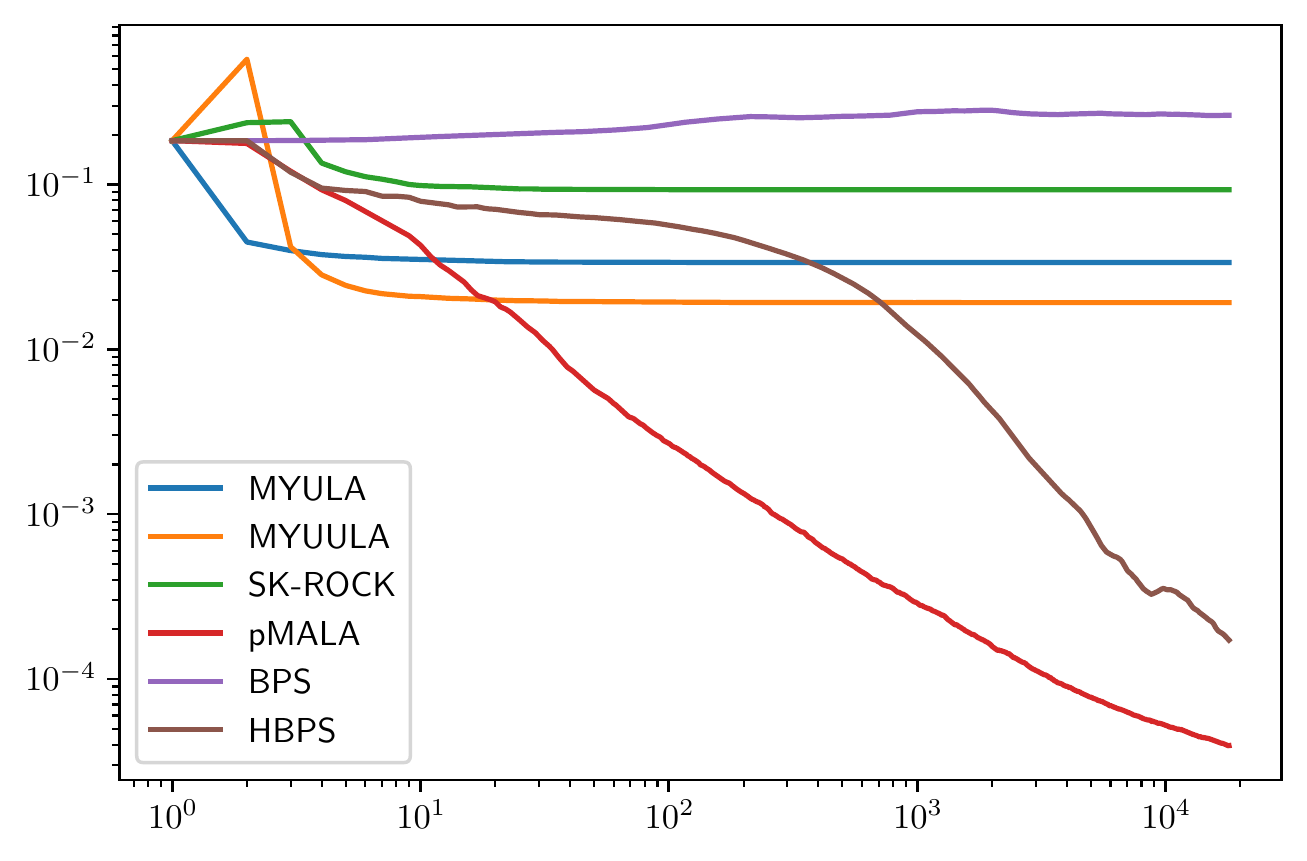}
  \end{minipage}
  \hfill  
  \begin{minipage}[b]{0.3\textwidth}
    \centering
    \includegraphics[scale=0.4]{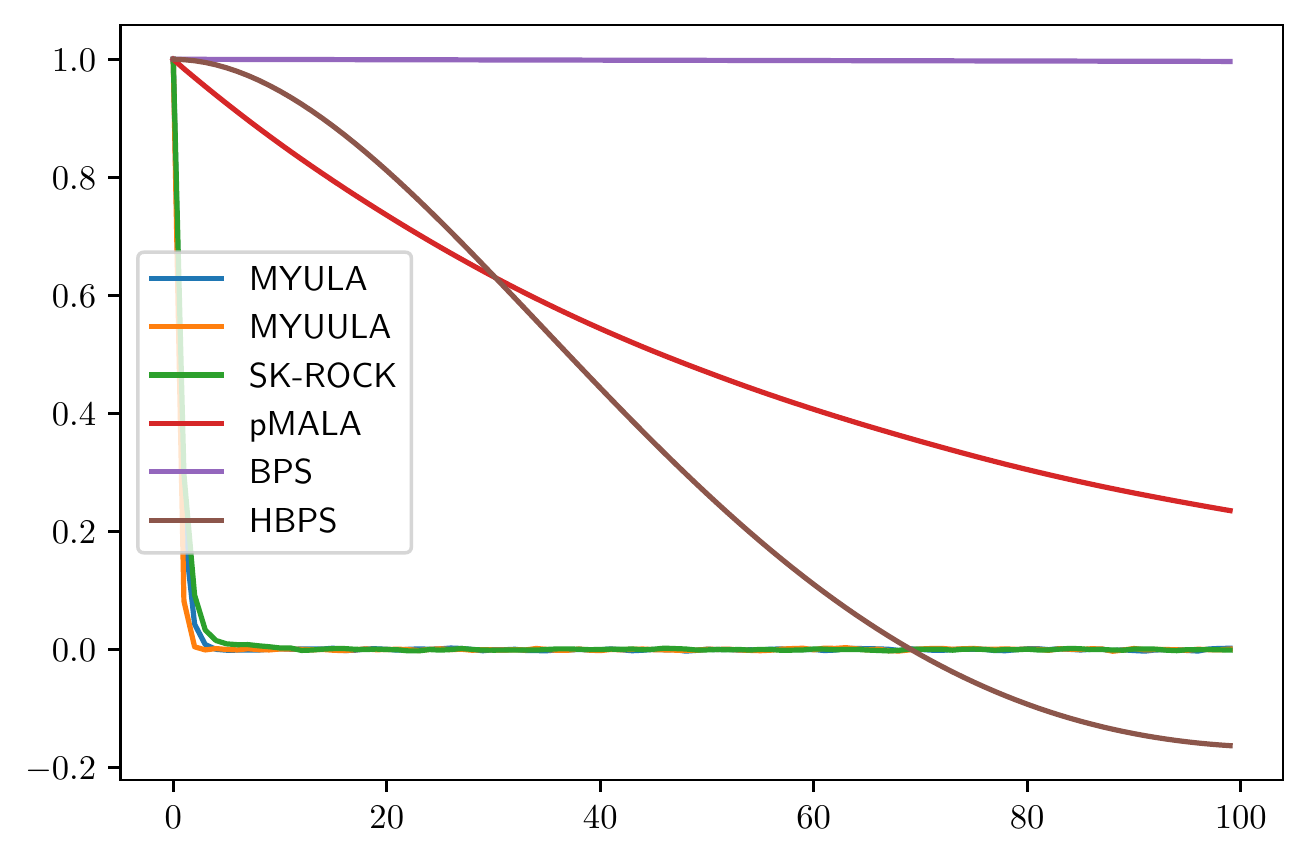}
  \end{minipage}
  \hfill
  \begin{minipage}[b]{0.3\textwidth}
    \centering
    \includegraphics[scale=0.4]{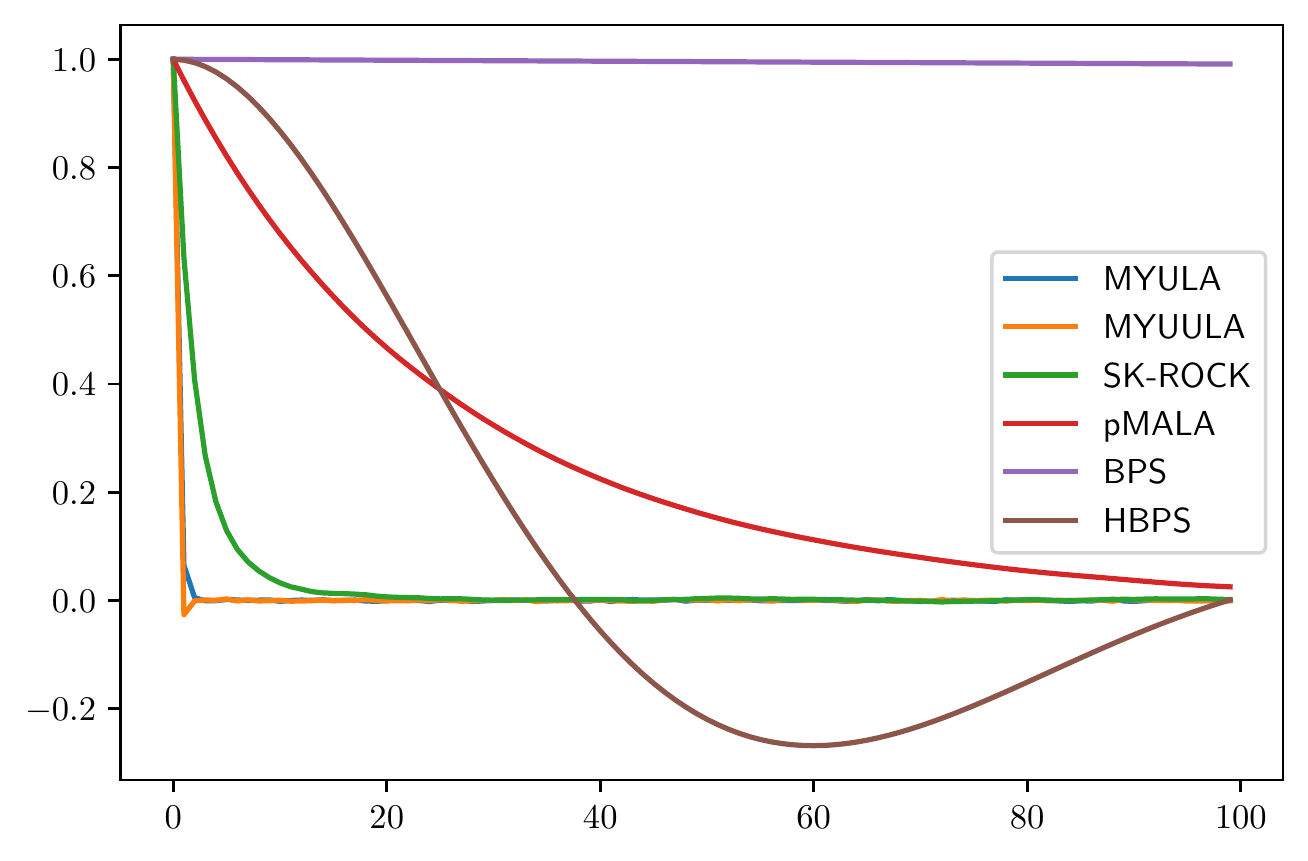}
  \end{minipage}
  \captionof{figure}{Results from the nuclear norm example. Left: MSE over time, for the different algorithms, run for half an hour each, on a log-log-scale. Middle: Autocorrelation for the slowest component, sample number adjusted for a fair comparison. Right: Autocorrelation for the fastest component, sample number adjusted for a fair comparison.\label{fig:CB_results}}
\end{minipage}

\subsection{Image Deblurring\label{ex:image_deblurring}}
Uncertainty quantification in images is generally a challenging computational problem, with samples from the posterior used to estimate credible intervals or provide model comparisons. We focus on a purely illustrative example involving the total variation prior similar to \cite[Example 4.1.2]{durmus2018efficient}. Let $x\in\mathbb R^{n_1\times n_2}$ be an image which we observe through $y=Hx+\xi$, where $H$ is a blurring operator that blurs a pixel $x_{i,j}$ uniformly with its closest neighbours ($5\times5$ patch), and $\xi\sim N(0,\sigma^2I_{n_1\times n_2})$. The log-prior is proportional to $-TV(x)=-\alpha\lVert \nabla_Dx\rVert_1$, where $\nabla_D$ is the two-dimensional discrete gradient operator as defined in \cite{chambolle2004algorithm}, and $\alpha$ is a fixed parameter. The application of the TV prior is common in a wide array of imaging applications, as it emphasizes smooth surfaces bounded by distinct edges. As the authors of \cite{durmus2018efficient} we chose the $256\times256$ "boat" test image, and set $\alpha=0.03$, $\sigma=0.47$. The posterior is given by
\begin{align}\label{ID_post}
    \pi(x)\propto \exp\left(-\frac{1}{2\sigma^2}\lVert Hx-y\rVert^2_2-\alpha TV(x)\right).
\end{align}
The TV-prior decomposes into a sum where each entry only depends on neighboring points; the uniform blur operator is similarly local. This implies in combination that the posterior can be factorized at granularities defined by the user,  and we can therefore apply the local BPS. We stress that the global BPS struggles in high dimensions \cite{deligiannidis2018randomized}, and thus localization is necessary for it to be a competitive algorithm in these settings. The proximal operator is not available in closed form for the TV-prior, and hence requires evaluation via numerical schemes such as the Douglas-Rachford algorithm introduced in \cite{douglas1956numerical} or the Chambolle-Pock algorithm \cite{chambolle2011first}. While these algorithms in general are efficient, they slow down significantly as the precision of the envelope is increased.

\begin{figure}
\begin{center}
	\includegraphics[scale=0.5]{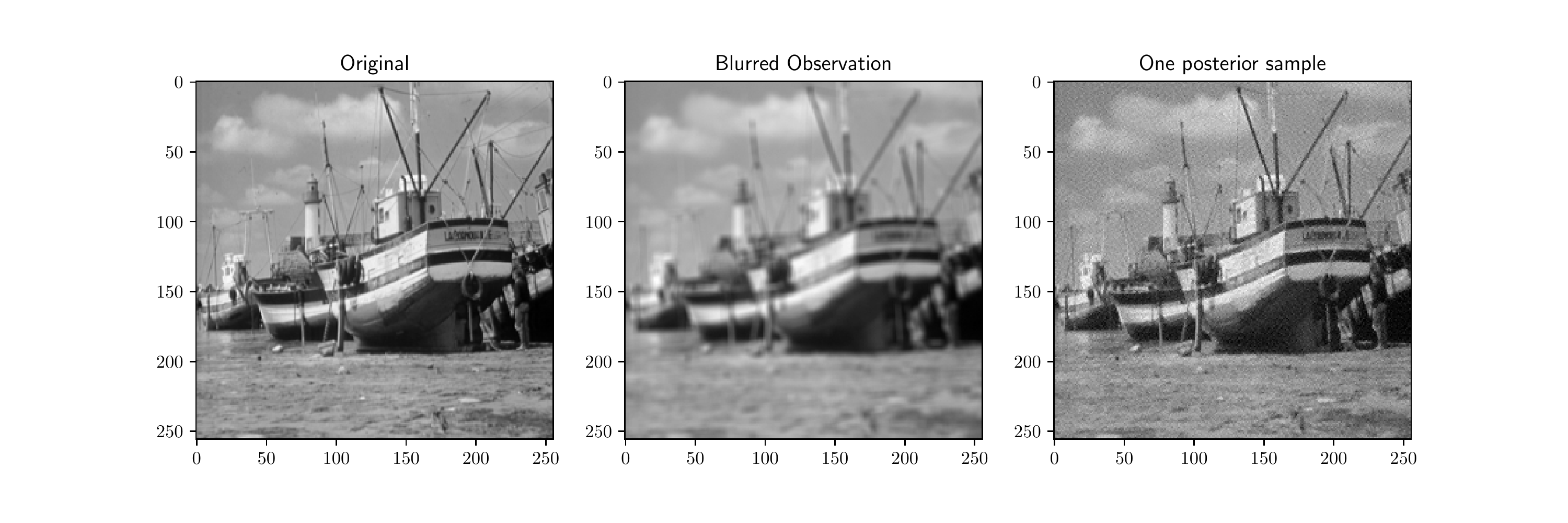}
	\caption{Left: The original $256\times256$ image. Center: The image after the application of the uniform blur operator. Right: A representative sample from the posterior distribution given in equation (\ref{ID_post}), obtained using the LBPS.}
	\label{fig:ID_images}
\end{center}
\end{figure} 

We compare the performances of the LBPS, the ZZS, pMALA, MY-ULA, MY-UULA, and SK-ROCK. For both the LBPS and the ZZS we estimated bounds on the prior- and likelihood-gradients, and used these constant bounds to generate computationally cheap events, avoiding any global evaluations of the gradient. For pMALA, we set $\lambda=2\delta=0.006$, giving us an acceptance ratio of $67\%$. For the last three samplers, we chose $\lambda=0.45$ following the guidance in \cite{durmus2018efficient}. The goal is to sample from the posterior distribution when observing a blurred image, see Figure \ref{fig:ID_images}. Figure \ref{fig:ID_results} shows the mean squared error (MSE) and the structural similarity index (SSIM) between the mean estimates of the various algorithms and the `true' mean, as estimated by a long run of an asymptotically exact algorithm. Notably, unlike MY-(U)ULA, pMALA, and SK-ROCK, which require the evaluation of the proximal operator (which is not localizable), the LBPS and ZZS can be sped up using parallelization techniques: the implementation we used applied global rates to avoid recalculating the full posterior gradient after every event, but one may calculate the factor gradients at hardly any extra computational cost if one calculates them in parallel.

\begin{minipage}{\textwidth}
  \begin{minipage}[b]{0.49\textwidth}
    \centering
    \includegraphics[scale=0.5]{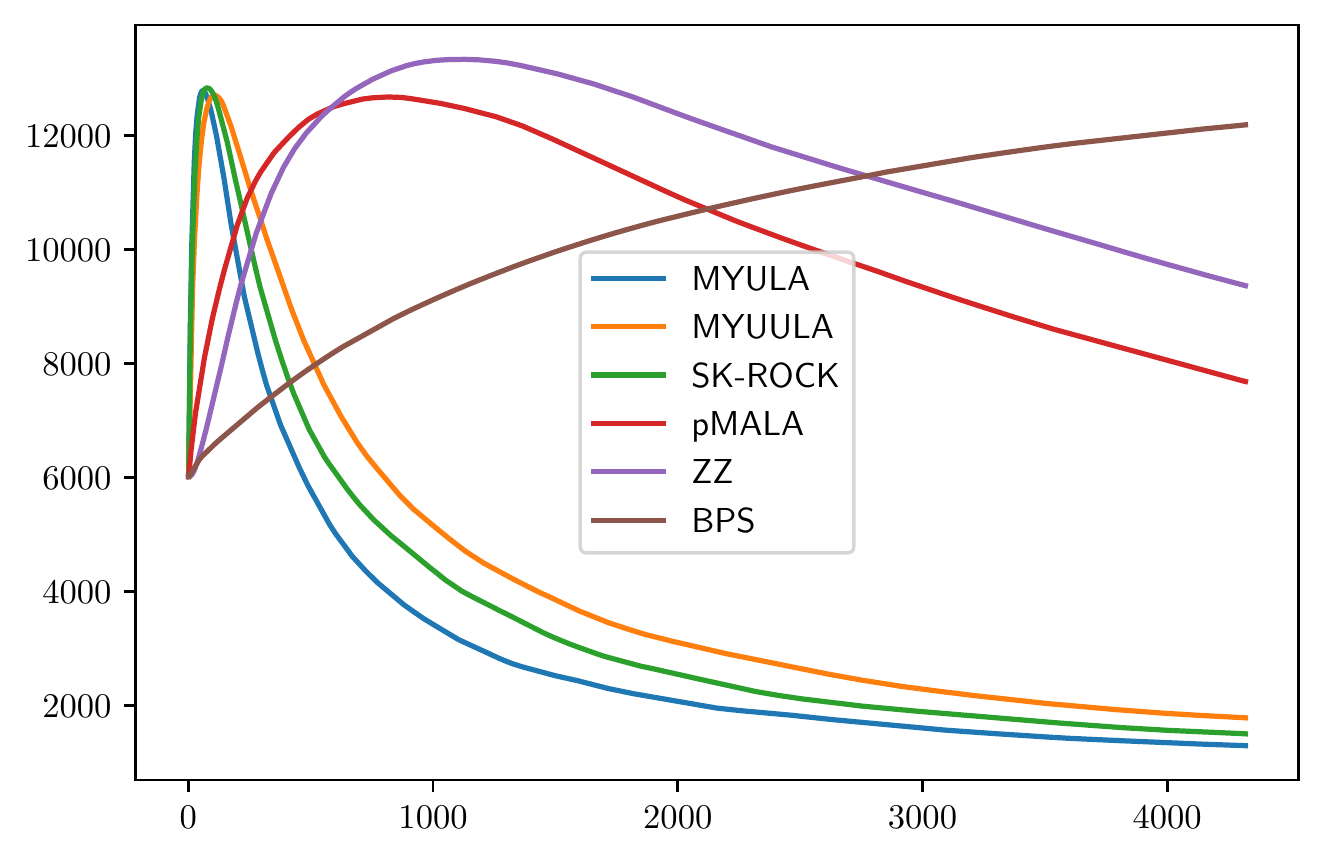}
  \end{minipage}
  \hfill
  \begin{minipage}[b]{0.49\textwidth}
    \centering
    \includegraphics[scale=0.5]{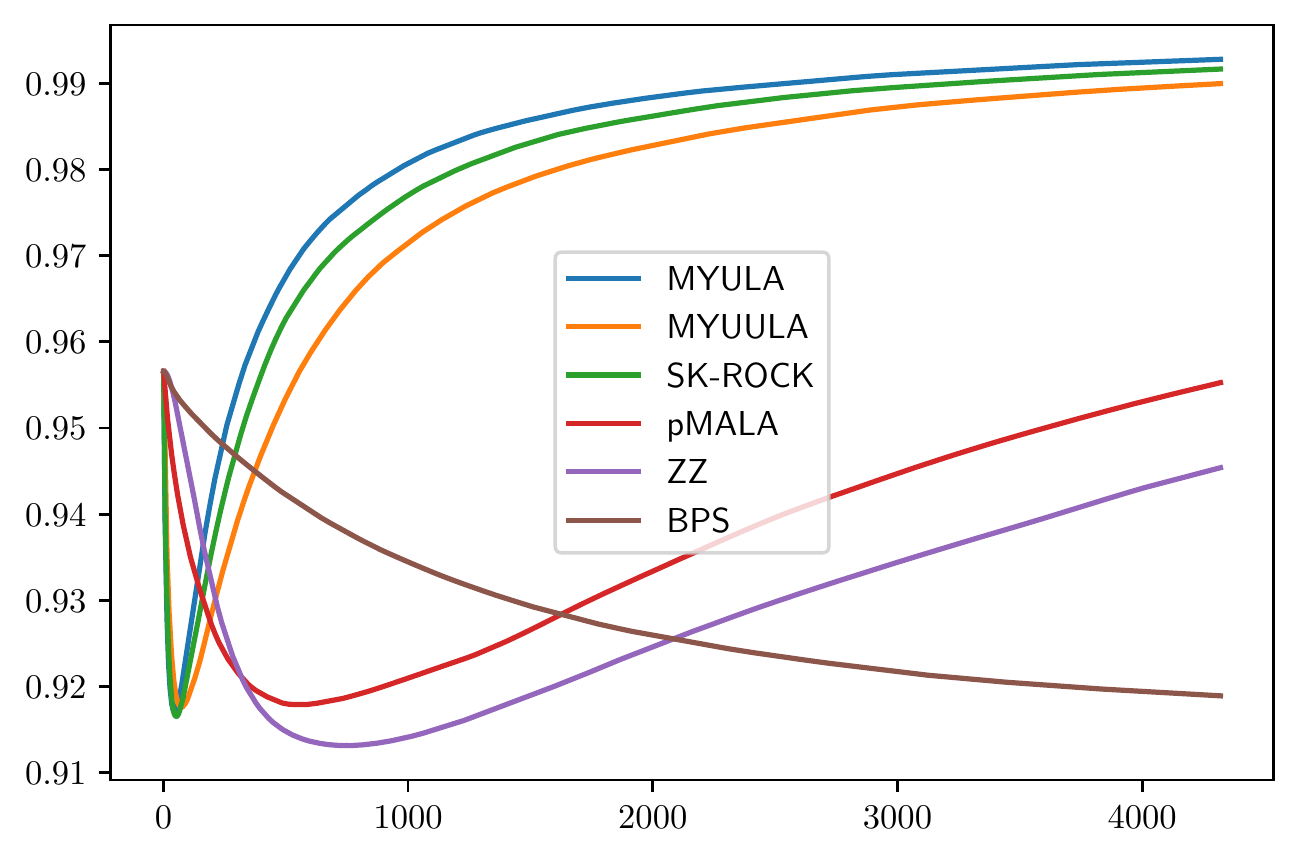}
    \end{minipage}
	\captionof{figure}{Results from the Image Deblurring example. Left: The MSE of the mean estimates, estimated every $10$ seconds. Right: The SSIM of the mean estimates, estimated every $10$ seconds.\label{fig:ID_results}}
\end{minipage}

\end{document}